\documentclass[a4paper, cleveref, autoref, thm-restate]{lipics-v2021}

\hideLIPIcs  

\newcommand{\paper}[1]{}
\newcommand{\report}[1]{#1} 

\title{Weighted Rewriting: Semiring Semantics for Abstract Reduction Systems}

\titlerunning{Weighted Rewriting: Semiring Semantics for Abstract Reduction Systems}

\author{Emma Ahrens}{RWTH Aachen University, Aachen, Germany \and \url{https://www.unravel.rwth-aachen.de/go/id/bdwvwa}}{ahrens@cs.rwth-aachen.de}{https://orcid.org/0000-0002-6394-3351}{}

\author{Jan-Christoph Kassing}{RWTH Aachen University, Aachen, Germany \and \url{https://jckassing.com}}{kassing@cs.rwth-aachen.de}{https://orcid.org/0009-0001-9972-2470}{}

\author{Jürgen Giesl}{RWTH Aachen University, Aachen, Germany \and \url{https://verify.rwth-aachen.de/giesl/}}{giesl@informatik.rwth-aachen.de}{https://orcid.org/0000-0003-0283-8520}{}

\author{Joost-Pieter Katoen}{RWTH Aachen University, Aachen, Germany \and \url{https://moves.rwth-aachen.de/people/katoen/}}{katoen@cs.rwth-aachen.de}{https://orcid.org/0000-0002-6143-1926}{}

\authorrunning{E.\ Ahrens, J.-C.\ Kassing, J.\ Giesl, and J.-P.\ Katoen}

\Copyright{Emma Ahrens, Jan-Christoph Kassing, Jürgen Giesl, and Joost-Pieter Katoen}

\paper{\relatedversiondetails{See \cite{SemiRingReport}. Full version, including all proofs}{https://arxiv.org/abs/??}}



\funding{funded by the DFG Research Training Group 2236 UnRAVeL}

\nolinenumbers

\newcommand{\disabledcmt}[1]{}
\newcommand{\oldcmt}[1]{}

\newcommand{\YG}[1]{\colorbox{white}{${#1}$}}

\newcommand{\IZ}{\mathbb{Z}}
\newcommand{\IN}{\mathbb{N}}
\newcommand{\IR}{\mathbb{R}}
\newcommand{\IS}{\mathbb{S}}

\newcommand{\bound}{C}
\newcommand{\multiSeq}[1]{\mathsf{Seq}(#1)}
\newcommand{\multiX}{\multiSeq{X}}
\newcommand{\multiA}{\multiSeq{A}}
\newcommand{\multiS}{\multiSeq{\IS}}
\newcommand{\maxVar}{\mathsf{max}\VSet}
\newcommand{\AggSet}{\mathcal{F}}
\newcommand{\FT}{\F{T}}
\newcommand{\NFto}{\NF_{\to}}
\newcommand{\fNF}{\mathsf{f}_{\NF}}
\newcommand{\aggrrule}{\mathsf{Aggr}_{\rulesemi}}
\newcommand{\aggr}[1]{\mathsf{Aggr}_{#1}}
\newcommand{\rulesemi}{a\to B} 

\newcommand{\semantics}[3]{\llbracket #1 \rrbracket^{#2}_{#3}}

\newcommand{\semelem}[1]{\semantics{#1}{}{}}
\newcommand{\treeto}{\Rightarrow}

\newcommand{\wARS}{(A, \to, \IS, \fNF, \aggrrule)}
\newcommand{\cplx}[1]{\operatorname{\mathsf{cplx}}(#1)}

\newcommand{\WP}{\mathsf{wp}}
\newcommand{\WLP}{\mathsf{wlp}}
\newcommand{\WPSem}[2]{\mathsf{wp}~\semelem{#1}~(#2)}

\newcommand{\semiplus}{\oplus}
\newcommand{\semiplusbig}{\bigoplus}
\newcommand{\semimult}{\odot}
\newcommand{\semimultbig}{\bigodot}
\newcommand{\seminull}{\mathbf{0}}
\newcommand{\semione}{\mathbf{1}}
\newcommand{\semifin}{T_{\textsf{fin}}}

\newcommand{\semitop}{\top}
\newcommand{\semibot}{\bot}
\newcommand{\semigeq}{\succcurlyeq}

\newcommand{\semileq}{\preccurlyeq}

\newcommand{\semisup}{\bigsqcup}
\newcommand{\semiinf}{\bigsqcap}

\newcommand{\semilong}{(S,\semiplus,\semimult,\seminull,\semione)}

\newcommand{\seminatext}{\IS_{\IN^{\infty}}}
\newcommand{\seminatextlong}{(\IN^{\infty}, +, \cdot, 0, 1)}
\newcommand{\semiconf}{\IS_{\mathsf{conf}}}
\newcommand{\semiconflong}{([0,1], \max, \cdot, 0, 1)}
\newcommand{\semitrop}{\IS_{\mathsf{trop}}}
\newcommand{\semitroplong}{(\IN^{\infty}, \min, +, \infty, 0)}
\newcommand{\semiarc}{\IS_{\mathsf{arc}}}
\newcommand{\semiarclong}{(\IN^{\pm\infty}, \max, +, -\infty, 0)}
\newcommand{\semilang}{\IS_{\Sigma}}
\newcommand{\semilanglong}{(2^{\Sigma^*}, \cup, \cdot, \emptyset, \{\varepsilon\})}
\newcommand{\semibool}{\IS_{\mathbb{B}}}
\newcommand{\semiboollong}{(\{\tfalse,\ttrue\}, \vee, \wedge, \tfalse, \ttrue)}
\newcommand{\semibottle}{\IS_{\mathsf{bottle}}}
\newcommand{\semibottlelong}{(\IR^{\pm\infty}, \max, \min, -\infty, \infty)}

\newcommand{\semirealext}{\IS_{\IR^{\infty}}}
\newcommand{\semirealextlong}{(\IR^{\infty}_{\geq 0}, +, \cdot, 0, 1)}

\def\namedlabel#1#2{\begingroup
    #2%
    \def\@currentlabel{#2}%
    \phantomsection\label{#1}\endgroup
}

\renewcommand{\emph}[1]{\index{#1}\textit{#1}}

\renewcommand{\emptyset}{\varnothing}

\newcommand{\F}[1]{\mathfrak{#1}}

\makeatletter
\def\moverlay{\mathpalette\mov@rlay}
\def\mov@rlay#1#2{\leavevmode\vtop{%
   \baselineskip\z@skip \lineskiplimit-\maxdimen
   \ialign{\hfil$\m@th#1##$\hfil\cr#2\crcr}}}
\newcommand{\charfusion}[3][\mathord]{
    #1{\ifx#1\mathop\vphantom{#2}\fi
        \mathpalette\mov@rlay{#2\cr#3}
      }
    \ifx#1\mathop\expandafter\displaylimits\fi}
\makeatother

\newcommand{\Var}{\mathcal{V}}

\newcommand{\TSet}[2]{\mathcal{T}\left(#1,#2\right)}

\newcommand{\VSet}{\mathcal{V}}

\renewcommand{\O}{\mathcal{O}}
\newcommand{\R}{\mathcal{R}}
\newcommand{\W}{\mathcal{W}}

\newcommand{\Dist}{\operatorname{Dist}}

\newcommand{\Geo}{\operatorname{Geo}}

\newcommand{\PP}{\mathcal{P}}

\newcommand{\tplus}{\mathsf{plus}}
\newcommand{\ts}{\mathsf{s}}

\renewcommand{\O}{\mathcal{O}}

\newcommand{\ttrue}{\mathsf{true}}
\newcommand{\tfalse}{\mathsf{false}}

\newcommand{\tidle}{\mathsf{idle}}
\newcommand{\twait}{\mathsf{wait}}
\newcommand{\trun}{\mathsf{run}}

\crefname{definition}{Def.}{Def.}
\crefname{example}{Ex.}{Ex.}
\crefname{counterexample}{Counterex.}{Counterex.}
\crefname{appendix}{App.}{App.}
\crefname{ex}{Ex.}{Ex.}
\crefname{theorem}{Thm.}{Thm.}
\crefname{lemma}{Lemma}{Lemmas}
\crefname{remark}{Rem.}{Rem.}
\crefname{section}{Sect.}{Sect.}
\crefname{subsection}{Sect.}{Sect.}
\crefname{subsubsection}{Sect.}{Sect.}
\crefname{line}{Line}{Lines}
\crefname{corollary}{Cor.}{Cor.}
\crefname{figure}{Fig.}{Fig.}
\crefname{enumi}{}{}
\crefname{algorithm}{Alg.}{Alg.}

\makeatletter
\NewDocumentCommand{\dparrow}{+O{} +O{0.5cm}}{%
\begin{tikzpicture}[baseline=-0.5ex] {
\node[inner sep=0](@1) at (-0,0) {};
\node[inner sep=0](@2) at (#2,0) {};
\draw [arrows={-Triangle[open]},shorten >= 1pt,shorten <= 1pt](@1)--(@2) node[pos=.5,above,inner sep=1pt] {\ensuremath{#1}};}
\end{tikzpicture}\xspace\xspace
}

\NewDocumentCommand{\myto}{+O{} +O{0.5cm}}{%
\begin{tikzpicture}[baseline=-0.5ex] {
\node[inner sep=0](@1) at (0,0) {};
\node[inner sep=0](@2) at (#2,0) {};
\draw [arrows={-to}](@1)--(@2) node[pos=.5,above,inner sep=1pt] {\ensuremath{#1}};}
\end{tikzpicture}\xspace
}

\NewDocumentCommand{\paraarrow}{+O{} +O{0.4cm}}{%
\begin{tikzpicture}[baseline=-0.5ex] {
\node[inner sep=0](@1) at (0,0) {};
\node[inner sep=0](@2) at (#2,0) {};
\node[inner sep=0](@3) at (0.07,0) {};
\draw [arrows={-to}](@1)--(@2) node[pos=.5,above,inner sep=1pt] {\ensuremath{#1}};
\draw [arrows={-to}](@1)--(@3);}
\end{tikzpicture}\xspace
}

\NewDocumentCommand{\paradparrow}{+O{} +O{0.4cm}}{%
\begin{tikzpicture}[baseline=-0.5ex] {
\node[inner sep=0](@1) at (0,0) {};
\node[inner sep=0](@2) at (#2,0) {};
\node[inner sep=0](@3) at (0.07,0) {};
\draw [arrows={-Triangle[open]}](@1)--(@2) node[pos=.5,above,inner sep=1pt] {\ensuremath{#1}};
\draw [arrows={-to}](@1)--(@3);}
\end{tikzpicture}\xspace
}

\newcommand{\oset}[2]{%
  {\mathop{#2}\limits^{\vbox to 1\ex@{\kern-\tw@\ex@
   \hbox{\scriptsize #1}\vss}}}}

\newcommand{\osetthree}[2]{%
  {\mathop{#2}\limits^{\vbox to 3\ex@{\kern-\tw@\ex@
   \hbox{\scriptsize #1}\vss}}}}

\newcommand{\osetfive}[2]{%
  {\mathop{#2}\limits^{\vbox to 5\ex@{\kern-\tw@\ex@
   \hbox{\scriptsize #1}\vss}}}}

\newcommand{\osetminus}[2]{%
  {\mathop{#2}\limits^{\vbox to -2\ex@{\kern-\tw@\ex@
   \hbox{\scriptsize #1}\vss}}}}
\makeatother

\newcommand{\mto}{\mathrel{\smash{\stackrel{\raisebox{2.0pt}{\tiny $\mathtt{s}\:$}}{\smash{\rightarrow}}}}}

\renewcommand{\phi}{\varphi}
\renewcommand{\emptyset}{\varnothing}
\newcommand{\NF}{\mathtt{NF}}

\newcommand{\RR}{\mathbb{R}}

\newcommand{\tored}[3]{
  \mathrel{
    \xhookrightarrow{{}_{\scriptstyle #1}}
    \!\!{}^{#2}_{#3}
  }
}

\newcommand{\defemph}[1]{{\rm #1}}

\usepackage{tikz}
\usetikzlibrary{shapes,calc,arrows,automata,decorations.pathmorphing,backgrounds,arrows.meta,shapes.geometric,scopes}
\RequirePackage{makecell}
\usepackage{nicefrac,xfrac}
\usepackage{mathtools}
\usepackage{stmaryrd}
\usepackage{amsmath,amssymb}
\usepackage{wrapfig}
\usepackage{xspace}
\usepackage{stmaryrd}
\usepackage[ruled, vlined]{algorithm2e}
\usepackage{eurosym}

\usepackage{babel,tabularx,ragged2e,booktabs}
\newcolumntype{L}{>{\RaggedRight\hspace{0pt}}X} 

\begin{document}
\allowdisplaybreaks

\maketitle 

\begin{abstract}
    We present novel semiring semantics for abstract reduction systems (ARSs).
    More precisely, we provide a weighted version of ARSs, where the reduction steps
    induce weights from a semiring.
    Inspired by provenance analysis in database theory and logic,
    we obtain a formalism that can be used for provenance analysis of arbitrary ARSs.
    Our semantics handle (possibly unbounded) non-determinism and possibly infinite reductions.
    Moreover, we develop several techniques to prove upper and lower bounds on
    the weights resulting from our semantics, 
    and show that in this way one obtains a uniform approach to
    analyze several different properties like termination, derivational
    complexity, space complexity, safety,
    as well as combinations of these properties. 
\end{abstract}


\section{Introduction}\label{Introduction}

Rewriting is a prominent formalism in computer science and notions like
termination, complexity, and confluence have been studied for decades for
abstract reduction systems (ARSs). 
Moreover, typical problems in computer science like evaluation of a
database query or a logical formula can be represented as a reduction system.

In this paper, we 
tackle the question whether numerous different analyses in the area of
rewriting, computation, logic, and deduction are ``inherently similar'', i.e., whether 
they can all be seen as special instances of a uniform framework. Analogous research
has been done, e.g., in the  database and logic
community, where there have been numerous approaches to
\emph{provenance analysis}, i.e., 
to not just analyze satisfiability but also explainability of certain results
(see, e.g., \cite{CheneyCT09,dannert2019ProvenanceAnalysisPerspective,tannen2007ProvenanceSemirings}). 
In these works, semirings are often used to obtain information
beyond just satisfiability of a query. For instance, they may be used to compute the
confidence in an answer or to calculate the cost of proving satisfiability.

\begin{example}[Provenance Analysis in Databases]\label{ex:boolean-formula}
    Consider two tables $R$, $P$ in a database over the universe
    $U = \{a,b\}$ with
    $R = U$ and $P = U\times U$ (where each atomic fact, such as $Ra$ or $Pab$, 
    has a cost in $\IN \cup \{\infty\}$),
    and the formula $\psi$ which represents a database query.
    
    \begin{center}
        $R =$ \begin{tabular}{ c|c } 
            & cost\\
            \hline
            $a$ & $2$ \\ 
            $b$ & $\infty$ \\ 
           \end{tabular} \qquad
        $P =$ \begin{tabular}{ cc|cccc|c } 
            & & cost & & & &cost\\
            \hline
            $a$ & $a$ & $2$ & &
            $b$ & $a$ & $\infty$ \\             
            $a$ & $b$ & $7$ & &
            $b$ & $b$ & $10$ \\ 
           \end{tabular} \qquad
        $\psi = Ra \wedge (Pab \vee Pbb)$
    \end{center}

    Our aim is to calculate the maximal cost of proving
    $\psi$. For the alternative use of information,
    i.e., for disjunctions $\psi
    = \phi \vee \phi'$, we take the maximal cost of proving that $\phi$ or
    $\phi'$ are satisfied. For the joint use of information, i.e., for conjunctions
    $\psi = \phi \wedge \phi'$, we take the sum of costs for proving that
    $\phi$ and $\phi'$ hold. We can use the \emph{arctic semiring}
    $\semiarc = \semiarclong$ to formalize this situation (i.e., we consider $\IN \cup \{
    -\infty, \infty \}$ with the operations $\max$ and $+$ whose identity elements are
    $-\infty$ and $0$,  respectively):
    \smallskip
    \begin{itemize}
        \item For an atom $\alpha$, we set $\semelem{\alpha} = c$, where $c$ is the cost
          of the atom $\alpha$.\footnote{We assume that our formulas do not use negation,
          a typical restriction for semiring semantics for logic.} 
        
        \item For formulas $\phi, \phi'$, we set $\semelem{\phi \vee
        \phi'} = \max\{\semelem{\phi}, \semelem{\phi'}\}$ and $\semelem{\phi \wedge
        \phi'} = \semelem{\phi} + \semelem{\phi'}$.
    \end{itemize}
    \smallskip
    This leads to a maximal cost of $\semelem{\psi} = \semelem{Ra} + \semelem{Pab \vee Pbb} =
    2 + \max\{7, 10\} = 12$ for proving $\psi$. 
    We can use a different semiring to calculate a different property. 
    For example, the confidence that a formula $\psi$ holds is computable via the
    \emph{confidence semiring} $\semiconf = \semiconflong$, where all
    atomic facts are given a certain confidence score.
\end{example}

In general, to compute the weight $\semelem{a}$ of an object $a$, 
one defines an \emph{interpretation of the facts or atoms} 
(e.g., the definition of $\semelem{\alpha}$) 
which maps objects to elements of the semiring, and
an \emph{aggregator function} for each reduction
step (e.g., the rules for $\semelem{\phi \wedge \phi'}$ and $\semelem{\phi \vee
  \phi'}$ above)
which operates on elements of the semiring.

In this work we generalize the idea of evaluating a database query
within a given semiring to ARSs. As usual, an ARS is a set $A$ together with a binary relation $\to$ denoting reductions. 
Note that we have to allow reductions from a single object to multiple
ones, as we may have to consider multiple successors (e.g., $\phi$ and $\phi'$ for the
formula $\phi \wedge \phi'$) in order to define the weight of $\phi \wedge \phi'$,
whereas in classical ARSs, objects are reduced to single objects.
This leads to the notion of \emph{sequence ARSs}. 
Similar ideas have been used for probabilistic ARSs \cite{avanzini2020probabilistic,BournezRTA02,bournez2005proving},
where a reduction relates a single object to a
multi-distribution over possible results. 
We will see that probabilistic ARSs can indeed also be expressed using our formalism.

\begin{example}[Provenance Analysis for ARSs]\label{ex:boolean-formula-as-mARS}
    The formulas from \Cref{ex:boolean-formula} fit into the concept of
    sequence ARSs: The set $A$ contains all propositional, negation-free formulas
    over the atomic facts $\NFto = \{ R u_1, P u_1 u_2 \mid u_1, u_2 \in U \} \subset A$ (the normal forms of the relation $\to$) and the relation $\to$ 
    is defined as $\phi \land \psi \to [\phi, \psi]$ 
    and $\phi \lor \psi \to [\phi, \psi]$, 
    where $[\phi, \psi]$ denotes the sequence 
    containing $\phi$ as first and $\psi$ as second element.
    Given a semiring,
    aggregator functions for the reductions steps, and an interpretation of
    the normal forms, we calculate the weight of a formula as in the previous example.
    See \Cref{Semiring Semantics ARS} for the formal definition.
\end{example}

In order to handle arbitrary ARSs, we have to deal with non-terminating reduction sequences
and with (possibly unbounded) non-determinism.
For that reason, to ensure that our semantics are well defined,
we consider semirings where the natural order 
(that is induced by addition of the semiring)
forms a complete lattice.

In most applications of semiring semantics in logic \cite
{dannert2019ProvenanceAnalysisPerspective,tannen2007ProvenanceSemirings}, a
higher ``truth value'' w.r.t.\ the natural order is more desirable, e.g., for
the confidence semiring $\semiconf$ one would like to obtain a value close to
the most desirable confidence $1$. However, in the application of software
verification, it is often the reverse, e.g., for computing the runtime of a
reduction or when considering the costs as in \Cref{ex:boolean-formula} for the arctic semiring $\semiarc = \semiarclong$.
While every weight $s < \infty$ may
still be acceptable, the aim is to prove \emph{boundedness}, i.e., that the
maximum $\infty$ (an infinite cost) cannot occur. For example, boundedness can
imply termination of the underlying ARS, it can ensure that certain bad states
cannot be reached (safety), etc. By considering tuples over different
semirings, we can combine multiple analyses 
into a single combined framework
analyzing combinations of properties, where simply performing each analysis separately fails.
In \Cref{Proving Upper Bounds}, we give sufficient conditions for boundedness 
and show that the interpretation method, 
a well-known method to, e.g., prove termination \cite{baader_nipkow_1998,lankford1979proving,terese2003term} 
or complexity bounds \cite{bonfanteAlgorithmsPolynomialInterpretation2001, bonfante2011QuasiinterpretationsWayControl, hofbauerTerminationProofsLength1989} of ARSs, 
can be generalized to a sound 
and (regarding continuous complete lattice semirings) complete technique to prove boundedness.
On the other hand, it is also of
interest to prove worst-case lower bounds on weights in order to find bugs or potential attacks 
(e.g., inputs which lead to a long runtime). Moreover, if one determines both worst-case upper and
lower bounds, this allows to check whether the bounds are (asymptotically) exact.
Thus, we present a
technique to find counterexamples with maximal weight in \Cref{Proving Lower Bounds}.

So in this paper, we define a uniform framework in the form of \emph{weighted ARSs}, 
whose special instances correspond to different semiring interpretations.
This indeed shows the simila\-rity between numerous analyses in
rewriting, computation, logic, and deduction, because they\linebreak[3]
can all be represented within our new  framework.
In particular, this uniform
framework allows\linebreak[3] to adapt techniques which were developed for one special instance 
in order to use them for other\linebreak[3] special instances.
While the current paper is 
a first (theoretical) contribution in this direction, 
eventually this may also improve the
automation of the analysis for certain instances.

\medskip

\textbf{Main Results of the Paper:}
\begin{itemize}
    \item We generalize abstract reduction systems to a weighted version (\Cref{def:wARS})
      that is powerful enough to express complex notions like termination, complexity, and safety of (probabi\-listic) rewriting, and even novel combinations of such properties (see \Cref{Expressability}).
    \item We provide several sufficient criteria that ensure boundedness
        (\Cref{thm:guaranteed-bounded,thm:guaranteed-bounded1}).
    \item Moreover, we give a sound (and complete in case of continuous semirings) 
      technique based\linebreak[3] on the well-known interpretation method
      to prove boundedness, i.e., to show that the weight of every object in 
      the ARS is smaller than the maximum of the semiring
      (\Cref{thm:boundedness}).
    \item Finally, we develop techniques to approximate the weights (\Cref{thm:approx})
      and to detect counter\-examples that show unboundedness (\Cref{thm:loops}).
\end{itemize}

\textbf{Related Work:} Semirings are actively studied in the database and the
 logic community. See \cite{tannen2007ProvenanceSemirings} for the first paper
 on \emph{semiring provenance} and \cite
 {glavic2021DataProvenance, green2017SemiringFrameworkDatabase} for further
 surveys. Moreover, a uniform framework via semirings has been developed
 in the context of \emph{weighted automata}~\cite{droste2009HandbookWeightedAutomata},
 which has led to a wealth of extensions and practical applications, 
 e.g., in digital image compression and model checking.
 There is also work on semiring semantics for declarative languages
 like, e.g., Datalog \cite{KhamisNPSW24}, which presents properties of the
 semiring that ensure upper\linebreak[3] bounds on how fast a Datalog program can be
 evaluated. Semiring semantics for the lambda calculus have been provided
 in \cite{laird2013WeightedRLambda}. In \cite{BELLE2020181}, a declarative
 programming framework was presen\-ted which unifies the analysis of different
 weighted model counting problems. Within software\linebreak[3] verification, semirings have
 been used in \cite{BatzGKKW22} for a definition of 
\emph{weighted imperative programming}, a Hoare-like semantics, and a
 corresponding weakest (liberal) precondition semantics, and extended to Kleene algebras with tests in \cite{Sedlar23}. The weakest
 precondition semantics of \cite{BatzGKKW22} can also be expressed in our formalism, see
\Cref{WP Comparison}. 
For ARSs, so far only costs in specific semirings have been considered, 
e.g., in \cite{avanzini2020probabilistic,avanziniModularCostAnalysis2020,kop2023CostSizeSemanticsCallByValue,naaf2017ComplexityAnalysisTerm}.
Compared to all this related work, we present the first general
semiring semantics for abstract reduction systems and demonstrate how to use
semirings for analyzing different properties of programs in a unified way.

\medskip

\textbf{Structure:}
We give some preliminaries on abstract reduction systems and semirings in \Cref{Preliminaries}.
Then, we introduce the new notion of \emph{weighted ARSs} in \Cref{Semiring Semantics ARS} that defines semiring semantics for \emph{sequence ARSs}.
We illustrate the expressivity and applicability of this formalism in \Cref{Expressability}.
In \Cref{Proving Upper Bounds}, we show how to prove boundedness. 
Here, we first give sufficient criteria that guarantee boundedness,
and then we introduce the interpretation method for proving boundedness in general.
In \Cref{Proving Lower Bounds}, we discuss the problem of finding a worst-case lower bound on the weight, 
or even a counterexample, 
i.e., we present a technique to find reduction sequences that lead to unbounded weight.
In \Cref{Conclusion}, we conclude and discuss ideas for future work.
Finally, in \Cref{WP Comparison} we show how to express the semantics of \cite{BatzGKKW22} 
in our setting, \report{and all proofs can be found in 
\Cref{Additional Theory}.}\paper{and we refer to \cite{SemiRingReport} for all proofs.}


\section{Preliminaries}\label{Preliminaries}

In this section, we give a brief introduction to abstract reduction systems
and define complete lattice semirings that will ensure well-defined
semiring semantics for ARSs in \Cref{Semiring Semantics ARS}.
The following definition of ARSs adheres to the
commonly used definition, see, e.g., \cite{baader_nipkow_1998,terese2003term}.

\begin{definition}[Abstract Reduction System, Normal Form, Determinism]
    An \defemph{abstract reduction system (ARS)} is a pair $(A, \to)$ with
    a set $A$ and a binary relation $\to\; \subseteq A \times A$.
    \begin{itemize}
        \item An object $a\in A$ \defemph{reduces} to $b$ in a single step,
        abbreviated as $a \to b$, if $ (a,b) \in\; \to$.

        \item An object $a \in A$ is called a \defemph{normal form} if there is
        no $b \in A$ such that $a \to b$. 
        The set of all normal forms for $(A,\to)$ is denoted by $\NFto$.

        \item Finally, the ARS $(A, \to)$ is \defemph{deterministic} if for
        every object $a \in A$ there exists at most one $b \in A$ with
        $a \to b$.
        It is \defemph{finitely non-deterministic}\footnote{When only regarding
        classical ARSs, this notion is often called
        ``finitely branching'' instead. However, since we will regard \emph{sequence} ARSs
        in \Cref{def:multiARS}, in this paper the notion ``finitely branching'' will refer to the branching of their
        reduction trees, see \Cref{def:Reduction Tree}.}  if for every object $a\in A$
        there exist at most finitely many objects $b\in A$ with $a \to b$.
    \end{itemize}
\end{definition}

An important property of ARS is termination, i.e., the absence of infinite behavior.

\begin{definition}[Reduction Sequence, Termination]
    Let $(A,\to)$ be an ARS.
    A \defemph{reduction sequence} is a finite or infinite sequence
    $a_1 \to a_2 \to \cdots$ with $a_i \in A$, and we say that $(A,\to)$ is \defemph{terminating} if there exists
    no infinite reduction sequence.
\end{definition}

In our new notion of weighted rewriting, we \emph{weigh} 
the normal forms in $\NFto$ by elements of a semiring, 
which consists of a set $S$ associated with two operations $\semiplus$ and $\semimult$.

\begin{definition}[Semiring]
    A \defemph{semiring} $\IS$ is a tuple $\semilong$ consisting of a set $S$ (called the
    carrier) together with two binary functions $\semiplus, \semimult : (S \times S) \to S$
    such that $(S, \semiplus, \seminull)$ is a commutative monoid
    (i.e., $\semiplus$ is commutative and associative with identity element $\seminull$),
    $(S, \semimult, \semione)$ is a monoid, and $\semimult$ distributes over
    $\semiplus$. Furthermore, $\seminull$ is a multiplicative annihilator,
    i.e., $\seminull \semimult s = s \semimult \seminull = \seminull$ for all
    $s \in S$.
\end{definition}

Sometimes we write $\semiplus_{\IS}$ or $\seminull_{\IS}$ to clearly indicate the
semiring $\IS$. If it is clear from the context, we also use $\IS$ to denote
the carrier $S$.
In \Cref{Introduction}, we already mentioned some examples of semirings,
namely the confidence semiring $\semiconf$ and the arctic semiring $\semiarc$.
\Cref{table:semirings} lists some relevant semirings for this work. 
Here, the multiplication in the \emph{formal language semiring}
$\semilang$ is pairwise concatenation, i.e., for $P_1, P_2 \subseteq \Sigma^*$, we have
$P_1 \cdot P_2 = \{ uv \mid u \in P_1, v \in P_2 \}$.

Later, in \Cref{Proving Upper Bounds,Proving Lower Bounds}, 
we establish upper and lower bounds on the weight of a given reduction
sequence, respectively. Hence, we need an order on the elements in the
semiring. Additionally, the order should be defined in a way that guarantees
well-definedness of our semantics. We accomplish this by using the natural
order\footnote{One can also generalize our results to partially ordered
semirings, where the partial order is \emph{compatible} with addition and
multiplication, i.e., addition and multiplication are monotonic (see \Cref{def:monotonic}). Note that the natural order is the least (w.r.t.\ $\subseteq$)
such partial order.} induced by the addition $\oplus$.

\begin{figure}
    \small
    \centering
    {\setcellgapes{1.8pt}
    \makegapedcells
    \begin{tabular}{llll}
        $\seminatext$ &$= \seminatextlong$ &
        $\semirealext$ &$= \semirealextlong$ \\
        $\semitrop$ &$= \semitroplong$ &
        $\semiarc$ &$ = \semiarclong$ \\
        $\semibool$ &$ = \semiboollong$ &
        $\semiconf$ &$= \semiconflong$ \\
        $\semibottle$ &$= \semibottlelong$ &
        $\semilang$ &$= \semilanglong$ \!
    \end{tabular}}
    \caption{Non-exhaustive list of complete lattice semirings $\IS = \semilong$.}
    \label{table:semirings}
   \end{figure}

\begin{definition}[Natural Order]
    The \defemph{natural order} $\semileq$ on a semiring $\IS$ is defined for all
    $s, t\in \IS$ as $s \semileq t$ iff there exists an element $u \in \IS$ with
    $s \semiplus u = t$. A semiring is called \defemph{naturally ordered} if
    $\semileq$ is a partial order, i.e., reflexive, transitive, and
    antisymmetric.\footnote{\label{ReflexiveTransitive}By definition, $\semileq$ is reflexive
    since $s \semiplus \seminull = s$, 
    and transitive since $s \semiplus v = t$ and $t \semiplus w = u$ imply
    $s \semiplus (v \semiplus w) = u$ by associativity of $\semiplus$.
    So the only real requirement is antisymmetry.} 
\end{definition}

The lack of ``negative elements'' in semirings 
ensures that we can define the natural order,
because as soon as there exists an element different 
from $\seminull$ with an additive inverse, 
the relation $\semileq$ is not antisymmetric anymore. 
Every semiring in \Cref{table:semirings} is naturally ordered. 
If the additive operation is addition or maximum,
then the order corresponds to the usual order on the extended naturals or extended reals. 
The natural order on $\semilang$ is the subset relation 
(${\semileq_{\semilang}} = {\subseteq}$).
Moreover, since the additive operation in the \emph{tropical semiring} $\semitrop$ is the minimum, 
the natural order in this semiring is the reverse of the usual order.

Our semantics in \Cref{Semiring Semantics ARS} consider demonic non-determinism.
Hence, to analyze worst-case behavior, we want to take the least upper bound
over all (possibly uncountably many) schedulers to resolve all non-determinism,
i.e., we want to take the least upper bound of arbitrary (possibly uncountable)
sets. A partially ordered set, where the least upper bound exists for every two
elements, is called a \emph{join-semilattice}, see e.g., \cite
{Abramsky94}. Since we also need the existence of least upper bounds for
infinite uncountable sets, we require \emph{complete} lattices.\footnote
{Lattices are semilattices, where in addition to suprema also infima are
defined. While we do not need infima for our semantics, the existence of infima
is guaranteed if one assumes the existence of suprema for all
(possibly uncountable) subsets\report{, see \Cref{lem:semilattice} in \Cref
{Additional Theory}.}\paper{, see \cite{SemiRingReport}.}}

\begin{definition}[Complete Lattice] \label{def:complete-lattice}
    A naturally ordered semiring $\IS$ is a \defemph{complete lattice} 
    if the \defemph{least upper bound} (or \defemph{supremum})
    $\semisup T \in \IS$ exists for every set $T\subseteq \IS$.
\end{definition}

All semirings in \Cref{table:semirings} are naturally ordered and complete lattices. 
A complete lattice semiring does not only have a 
minimum\footnote{Already in naturally ordered semirings, 
$\seminull$ is the minimum: for all $s\in \IS$, we have $\seminull \semileq s$,
since $\seminull \semiplus s = s$.} $\seminull = \semibot
= \semisup \emptyset \in \IS$, 
but also a maximum $\semitop = \semisup S \in \IS$.
Furthermore, the existence of every supremum allows us to define infinite sums
and products of sequences, see, e.g.,
\cite{brinke2024SemiringProvenanceInfinite}.

We define sequences $T = (x_i)_{i\in I} = [x_1, x_2, \ldots] \subseteq \IS$,
where either $I = \{ i \in \IN \mid 1 \leq i \leq n\}$
for some $n \in \IN$
(then $T$ is a finite sequence of length $n$)
or $I = \IN_{\geq 1}$ (then $T$ is an infinite sequence). By $\multiX$, we denote the
set of all non-empty sequences over some set $X$. Furthermore, the subset relation
between two sequences $T = (x_i)_{i\in I}, T' = (x_i')_{i\in I'} \in \multiX$
is defined via prefixes, i.e., we write $T \subseteq T'$ if $I \subseteq I'$ and $x_i = x_i'$ for all $i\in I$.
For a finite sequence $T = [s_1,  \ldots, s_n] \in \multiS$, we
use the common abbreviation $\semiplusbig T = \semiplusbig_{i = 1}^{n} s_i =
s_1 \semiplus \cdots \semiplus s_n$ and $\semimultbig T = \semimultbig_{i =
1}^{n} s_i = s_1 \semimult \cdots \semimult s_n$.

\begin{definition}[Infinite Sums and Products] \label{def:infinite_operations}
    Let $\IS$ be a complete lattice semiring and $T$ an infinite sequence\footnote{Note that one typically defines sums for sets and not for sequences in provenance analysis. 
    However, we use the order given by the sequence for our semantics in \Cref{Semiring Semantics ARS}.} over $\IS$.
    Then we define the infinite sum and product of $T$ as

    \vspace*{-.4cm}

    {\small 
    \[\textstyle \hspace*{-.4cm} \semiplusbig T = \semisup \left\{ \semiplusbig \semifin \mid
    \semifin \text{ is a finite prefix of } T \right\},\quad 
    \semimultbig T = \semisup \left\{ \semimultbig \semifin \mid
    \semifin \text{ is a finite prefix of }  T \right\}.\]}\vspace*{-.6cm}
\end{definition}

Infinite sums and products are well defined, since the supremum of any set
exists in the complete lattice $\IS$, hence $\semiplusbig T \in \IS$ and
$\semimultbig T \in \IS$ for all infinite sequences $T$ over $\IS$.


\section{Semiring Semantics for Abstract Reduction Systems}\label{Semiring Semantics ARS}

In this section, we introduce \emph{weighted abstract reduction systems} 
by defining semiring semantics for ARSs 
and show that these semantics are well defined for complete lattice semirings. 
Our definitions are in line with provenance analysis as introduced in \Cref{ex:boolean-formula}. 
From now on, whenever we speak of a ``semiring'' we mean a complete lattice
semiring.

As for ordinary ARSs, we represent the relation $\to$ via rules which are selected
non-deter\-mi\-nis\-ti\-cally.
However, each rule can have multiple outcomes,
as in \Cref{ex:boolean-formula-as-mARS}. 
Syntactically, for a \emph{sequence abstract reduction system},
we use a relation $\to$ that relates a single object to a
sequence of all corresponding outcomes that we later
weigh using semiring elements. Note that sequences are necessary for reductions like
$\psi \land \psi \to [\psi,\psi]$, where both incarnations of $\psi$ may
reduce to different objects due to possible non-determinism.

\begin{definition}[Sequence Abstract Reduction System] \label{def:multiARS}
    A \defemph{sequence abstract reduction system} 
    (sARS) is a pair $(A, \to)$ consisting of a set $A$ 
    and a binary relation $\to\; \subseteq A \times \multiA$.
    \begin{itemize}
        \item We write $a \to B = [b_1, b_2, \ldots]$ if
        $B \in \multiA$ and $(a,B)\in\;\to$. Normal forms and
        $\NFto$ are defined as for ARSs.

        \item The sARS $(A, \to)$ is \defemph{deterministic} if for every $a \in A$ 
        there is at most one $B$ with $a \to B$ and  \defemph{finitely
        non-deterministic} if  
        for each $a \in A$, there are only finitely many $B
        \in \multiA$ with $a \to B$.
        The sARS is \defemph{finitely branching}
        if for every
        $a \to B$, the sequence $B$ is finite.
    \end{itemize}
\end{definition}

We have already seen an example of an sARS in \Cref{ex:boolean-formula-as-mARS}.
In ordinary ARSs, it suffices to consider reduction
sequences. When using sARSs, we obtain ordered reduction trees instead.

\begin{definition}[Reduction Tree]\label{def:Reduction Tree}
    Let $(A,\to)$ be an sARS.
    An \defemph{$(A,\to)$-reduction tree  ($(A,\to)$-RT)}
    $\FT\!=\!(V,E)$ is a labeled, ordered tree with nodes $V$ and directed edges $E\!\subseteq\!V\!\times\!V$, where
    \begin{itemize}
        \item every node $v\in V$ is labeled by an object $a_v \in A$ and
        \item every node $v$ together with its sequence of direct successors $vE = [w \in V \mid (v,w) \in E]$ either corresponds to a reduction step $a_v \to [a_w \mid w\in vE]$ or $vE$ is empty.
    \end{itemize}
    We say that $(A,\to)$ is \defemph{terminating} if all $(A,\to)$-RTs have finite
    depth.\footnote{One could
    instead attempt to define reduction trees in an inductive way. Then
    every node labeled with a value from $A$ would be a reduction tree and
    one could
    lift the reduction $\to$ to a binary relation $\treeto$ which
    extends reduction trees, i.e., 
    $\FT \treeto \FT'$ holds if there is a leaf $v$ of $\FT$
    and a reduction step $a_v \to B$ such that the reduction tree $\FT'$ extends $\FT$ by
    new leaves $w_b$ with $a_{w_b} = b$ and edges $(v,w_b)$ 
    for all $b\in B$. However, in this way one would only obtain reduction trees of finite
    depth, whereas we also need reduction trees of infinite depth in order
    to represent non-terminating reductions, which would require an additional limit step
    in the construction above.}
 \end{definition}

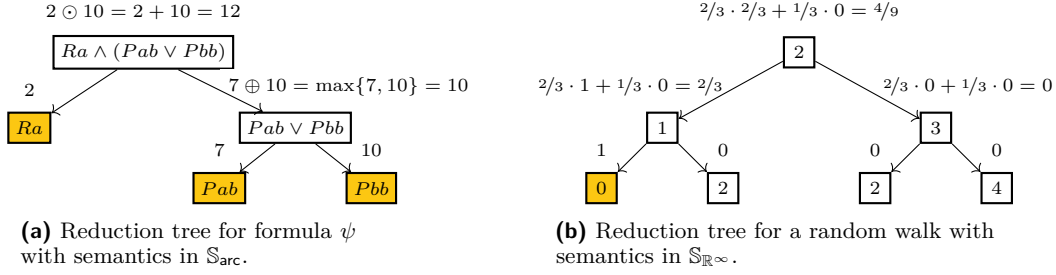
\begin{figure}
    \centering
    \begin{subfigure}{0.45\textwidth}
      \scriptsize \center \vspace*{-0.1cm}
      \hspace*{-.5cm}
      \begin{tikzpicture}
            \tikzstyle{myRect}=[thick,draw=black!100,fill=white!100,minimum size=4mm, shape=rectangle]
            \tikzstyle{trueRect}=[thick,draw=black!100,fill=lipicsYellow!100,minimum size=4mm, shape=rectangle]

            \node[myRect, label=above:{$\YG{2\semimult 10 = 2 + 10 = 12}$}] at (0, 0)
                (d) {$Ra \wedge (Pab \vee Pbb)$};

            \node[trueRect, label=above:{$\YG{2}$}] at (-1.5, -1)
                (dl) {$Ra$};
            \node[myRect, label=above:{$\qquad\qquad\quad\YG{7\semiplus 10 = \max\{7,10\} =
                  10}$}] at (2, -1) 
                (dr) {$Pab \vee Pbb$};
            \node[trueRect, label=above:{$\YG{7}$}] at (1, -1.8)
                (drl) {$Pab$};
            \node[trueRect, label=above:{$\YG{10}$}] at (3, -1.8)
                (drr) {$Pbb$};

            \draw (d) edge[->] (dl);
            \draw (d) edge[->] (dr);

            \draw (dr) edge[->] (drl);
            \draw (dr) edge[->] (drr);
      \end{tikzpicture}
      \subcaption{Reduction tree for formula $\psi$ \\
        with semantics in $\semiarc$.}
      \label{fig:red-tree-1}
    \end{subfigure}
    \hspace*{.5cm}
    \begin{subfigure}{0.45\textwidth}
      \scriptsize \center
      \hspace*{-.5cm}
      \begin{tikzpicture}
            \tikzstyle{myRect}=[thick,draw=black!100,fill=white!100,minimum size=4mm, shape=rectangle]
            \tikzstyle{trueRect}=[thick,draw=black!100,fill=lipicsYellow!100,minimum size=4mm, shape=rectangle]

            \node[myRect, label=above:{$\YG{\nicefrac{2}{3}\cdot\nicefrac{2}{3}
            + \nicefrac{1}{3}\cdot 0 = \nicefrac{4}{9}}$}] at (0, 0)
                (d) {$2$};

            \node[myRect, label=above:{$\YG{\nicefrac{2}{3}\cdot 1 + \nicefrac{1}{3}\cdot 0 = \nicefrac{2}{3}}\qquad\quad$}] at (-1.8, -1)
                (dl) {$1$};
            \node[myRect, label=above:{$\qquad\quad\YG{\nicefrac{2}{3}\cdot 0 + \nicefrac{1}{3}\cdot 0 = 0}$}] at (1.8, -1)
                (dr) {$3$};

            \node[trueRect, label=above:{$\YG{1}$}] at (-2.6, -1.8)
                (dll) {$0$};
            \node[myRect, label=above:{$\YG{0}$}] at (-1, -1.8)
                (dlr) {$2$};
            \node[myRect, label=above:{$\YG{0}$}] at (1, -1.8)
                (drl) {$2$};
            \node[myRect, label=above:{$\YG{0}$}] at (2.6, -1.8)
                (drr) {$4$};

            \draw (d) edge[->] (dl);
            \draw (d) edge[->] (dr);

            \draw (dl) edge[->] (dll);
            \draw (dl) edge[->] (dlr);
            \draw (dr) edge[->] (drl);
            \draw (dr) edge[->] (drr);
      \end{tikzpicture}
      \subcaption{Reduction tree for a random walk with\\ 
        semantics in $\semirealext$.}
      \label{fig:red-tree-2}
    \end{subfigure}
    \caption{Two example reduction trees, where each node $v$ is labeled with
    $a_v\in A$ and the small numbers are the corresponding weights $\semantics{\FT}{v}{}$.
    Colored nodes are labeled by normal forms.}
    \label{fig:reduction_trees}
\end{figure}

\Cref{fig:red-tree-1}
depicts a reduction tree for the formula $\psi$ from \Cref{ex:boolean-formula}
and \Cref{fig:red-tree-2} shows a reduction tree for a biased random walk starting at 2, see
\cref{ex:probabilistic-analysis}.
We define the weight of a reduction tree w.r.t.\ a
semiring $\IS$ by interpreting the leaf nodes as semiring elements and the
inner nodes as combinations of its children.
We use so-called \emph{aggregator functions} to combine 
weights occurring in reductions based on the semiring addition and multiplication.

\begin{definition}[Aggregator] \label{def:aggregatorFunction}
    Let $\IS$ be a semiring and $\VSet = \{ v_1, v_2, \ldots \}$ be a set of variables.
    Then the set of all \defemph{aggregators} $\AggSet$ (over $\IS$ and $\VSet$) 
    is the smallest set with
  \begin{itemize}
        \item $s \in \AggSet$ for every $s \in \IS$ (constants) and $v \in \AggSet$ for every $v \in \VSet$ (variables),

        \item $\semiplusbig F \in \AggSet$ (sums) and $\semimultbig F \in \AggSet$ (products) for every $F \in \multiSeq{\AggSet}$.
    \end{itemize}
    \noindent
    If an aggregator is constructed via finite sequences $F$,
    then it is a \defemph{finite aggregator}.
    Let $\VSet(f)$ be the set of all variables in $f \in \AggSet$ and let
    $\maxVar(f) = \sup\{i \mid v_i \in \VSet(f)\} \in \IN^{\infty}$.

    Let $f \in \AggSet$ and $n \in \IN^\infty$ with $n \geq \maxVar(f)$.
    Then the aggregator $f$ \defemph{induces a function}
      $f: \IS^n \to \IS$ in the obvious way: For a
    sequence $T = [s_1, \ldots] \in \multiSeq{\IS}$
    of length $n$, we have
    $s(T) = s$ for constants $s$, $v_i(T) = s_i$ for variables $v_i$ 
    and  $1 \leq i \leq n$, 
    and $(\bigcirc F)(T) = \bigcirc [f_1(T), f_2(T), \ldots]$
    for $F = [f_1, f_2, \ldots]$, where $\bigcirc \in \{\semiplusbig, \semimultbig\}$.
\end{definition}

Weighted rewriting considers an sARS $(A, \to)$ together with a semiring $\IS$,
and functions $\fNF$ and $\aggrrule$ in order to map objects from $A$ to elements of $\IS$.

\begin{definition}[Weighted Abstract Reduction System]\label{def:wARS}
    A tuple $\wARS$ is a \defemph{weighted abstract reduction system (wARS)} if \\[4pt]
\begin{minipage}[t]{0.35\textwidth}
    \begin{itemize}
        \item $(A, \to)$ is an sARS,
        \item $\IS$ is a semiring,
    \end{itemize}
\end{minipage} \hspace*{.5cm}
\begin{minipage}[t]{0.60\textwidth}
    \begin{itemize}
        \item $\fNF: \NFto \rightarrow \IS$ is the \defemph{interpretation of normal forms},
        \item $\aggrrule \in \AggSet$ 
          is the \defemph{aggregator} for every $\rulesemi$,
    \end{itemize}
\end{minipage}\\[4pt]
where $\maxVar(\aggrrule) \leq |B|$.
\end{definition}

For every $\aggrrule$, we consider the induced function
$\aggrrule: \IS^n \to \IS$  of arity $n = |B|$. 
Now we introduce our semiring semantics for reduction trees of finite depth
by using $\fNF$ to interpret
the leaves of a reduction tree that are labeled by normal forms. To interpret inner nodes,
we use aggregator functions that
distinguish between the different reductions. So
for the example from 
\Cref{ex:boolean-formula-as-mARS,fig:red-tree-1}, we use 
a function $\fNF$ where
$\fNF(\alpha)$ is the cost of atom $\alpha$, and 
we use $\aggr{\phi \wedge \psi \to
  [\phi, \psi]} = v_1 \semimult v_2$ for every rule $\phi \wedge \psi \to [\phi,\psi]$ and
$\aggr{\phi \vee \psi \to
[\phi, \psi]} = v_1 \semiplus v_2$ for every rule 
$\phi \vee \psi \to [\phi,\psi]$.
Thus, combining the weights of the children via aggregator functions enables us to
calculate a weight for the root of any reduction tree
with possibly infinite (countable) branching and finite depth.\footnote{Alternatively, one
could also consider finite-depth reduction trees as first-order ground terms (with function
symbols of possibly infinite arity). Then the semiring semantics of
\Cref{def:semiringSemantics} would correspond to 
a polynomial interpretation where the polynomials 
 $\aggr{a_v\to B}$
are constructed using the operations
$\semiplus$ and $\semimult$ of the semiring.}

\begin{definition}[Semiring Semantics] \label{def:semiringSemantics}
    For a wARS $\wARS$ and an 
    $(A,\to)$-RT $\FT=(V,E)$ of finite depth,
    we define the \defemph{weight}
    $\semantics{\FT}{v}{}$ of $\FT$ at node $v\in V$ as\linebreak[3]

\vspace*{-.52cm}
    
    {\small\begin{align*}
        \semantics{\FT}{v}{} &= \fNF(a_v) \; &&\text{if $a_v \in \NFto$}\\
        \semantics{\FT}{v}{} &= \seminull && \text{if $v$ is a leaf
and $a_v \notin \NFto$}\\   
        \semantics{\FT}{v}{} &= \aggr{a_v\to B}\left[ \semantics{\FT}{w}{} \mid w \in vE \right]
        &&\text{if $v$ is an inner node and } B = [a_w \mid w \in vE].
    \end{align*}}

    \vspace*{-.1cm}

    \noindent
      The weight of the whole RT $\FT$ is $\semantics{\FT}{}{}
    = \semantics{\FT}{r}{}$, where $r \in V$ is the root \pagebreak[3] node of $\FT$.    
\end{definition}

Note that the order of the children  $B = [a_w \mid w \in vE]$ is crucial when applying
$\aggr{a_v\to B}$ to $\left[\semantics{\FT}{w}{} \mid w \in vE \right]$.
This is needed, e.g., when an aggregator is used to
represent different probabilities for the successor objects in a biased random walk, see,
e.g.,
\Cref{fig:red-tree-2} and \Cref{Expressability:Probabilistic}.

A reduction tree
of finite depth represents one possible execution of the sARS up to a certain number of
steps, where the non-determinism is resolved by some fixed scheduler.
To define the semantics of an object $a \in A$, we consider any number of reduction steps
and all possible schedulers.
Then we define the weight of $a$ as  the least upper
bound of the weights of all finite-depth reduction trees whose root is labeled with $a$.

\begin{definition}[Semantics with Demonic Non-Determinism]\label{def:element-semiring-semantics}
  For a wARS $\W = (A, \to, \IS, \fNF,\linebreak[3] \aggrrule)$ and $a \in A$,  let
  $\Phi(a)$ be the set of all $(A,\to)$-reduction trees of finite depth
  whose root node is
  labeled with $a$. Then we define
  the \defemph{weight} of $a$ as
  $\semelem{a} = \semisup \{ \semantics{\FT}{}{} \mid \FT \in  \Phi(a)\}$.\footnote{In
  principle, $\semelem{a}$,  $\semantics{\FT}{}{}$, and $\semantics{\FT}{v}{}$
  are indexed by $\W$, 
  but we omitted this index for readability.}
\end{definition}

Due to possibly uncountably many schedulers, 
there might be uncountably many reduction trees each with a different weight\report{ (see
\Cref{lem:nondeterminism} in \Cref{Additional Theory}).}\paper{, see \cite{SemiRingReport}.}
Nevertheless, since $\IS$ is a complete lattice, the supremum
of every set exists and thus,
the weight of every object is well defined.

\begin{corollary}[Well-Defined Semantics]\label{thm:well-defined}
  For any wARS $\wARS$,
  the weight $\semelem{a}$ is well defined for every object $a\in A$.
\end{corollary}

The set $\Phi(a)$ takes on different shapes depending on the reduction system.
If the reduction system is deterministic, then $\Phi(a)$ consists of
all finite-depth prefixes of a single (potentially infinite-depth) tree.
If the sARS is non-deterministic,
then $\Phi(a)$ may contain uncountably many trees.
The maximal size of the sequences in the reduction rules determines
the maximal branching degree of the trees. If the sARS is finitely branching, so are the trees in
$\Phi(a)$.

The computation of the weight $\semelem{a}$ is undecidable in general, 
since computing single steps with $\to$ may already be undecidable.
However, even if the reductions $\rulesemi$, the interpretation of the normal forms
$\fNF$, and the aggregator functions $\aggrrule$ are computable,
computing the weight $\semelem{a}$ can still be undecidable, because it can express notions like
termination of deterministic systems as demonstrated in the
next section (\Cref{Expressability:Runtime}).


\section{Expressivity of Semiring Semantics}\label{Expressability}

In this section we give several examples to demonstrate 
the versatility and expressive power of our  new formalism,
and show that existing approaches for the analysis of reduction systems 
actually consider specific semirings.

\subsection{Termination and Complexity}\label{Expressability:Runtime}

We can extend any ARS $(A,\to)$ to a wARS 
$\cplx{A,\to} = (A, \mto, \seminatext, \fNF^{\mathsf{cplx}}, 
\mathsf{Aggr}_{a \mto B}^{\mathsf{cplx}})$
such that $\semelem{a}$ is equal to the supremum over 
the lengths of all reduction sequences starting in $a \in A$.\linebreak[3]
For this, we use the sequence relation $\mto \;= \{a \mto [b] \mid a \to b\}$,
the \emph{extended naturals semi\-ring} $\seminatext$,
the interpretation 
$\fNF^{\mathsf{cplx}}(a) = \seminull_{\seminatext} = 0$ for all $a \in \NF_{\mto}$,
and the aggregator $\mathsf{Aggr}_{a \mto [b]}^{\mathsf{cplx}}\linebreak[3]
 =  1 \semiplus_{\seminatext} v_1 = 1 + v_1$ whenever $a \mto [b]$.
Recall that aggregators use a fixed set of variables $\Var = \{v_1, \ldots\}$.
The derivational complexity of $(A,\to)$ 
(i.e., the supremum of the lengths of possible reduction sequences) 
is obtained by analyzing 
the weights $\semelem{a}$ of $\cplx{A,\to}$.

Techniques for automatic complexity and termination analysis have been developed for, e.g., term rewrite
systems (TRSs) in the literature \cite{baader_nipkow_1998,terese2003term}, and there is
an annual \emph{Termination and 
Complexity Competition} with numerous participating tools \cite{termcomp}. 
In term rewriting, \pagebreak[3] one considers terms $t \in \TSet{\Sigma}{\VSet}$ over
a set of function symbols $\Sigma$
and a set of variables $\VSet$. The reduction relation $\to_{\R}$ is defined via a set of rewrite rules $\R$:
If the left-hand side of a rewrite rule in $\R$ matches a subterm, we can replace this subterm 
with the right-hand side of the rewrite rule instantiated by the matching substitution.
For details, see, e.g., \cite{baader_nipkow_1998,terese2003term}.

\begin{example}[TRS for Addition]\label{ex:complexity-analysis}
    Let $\Sigma = \{\tplus\}$ and $\VSet = \{x,y\}$.
    A TRS $\R$ computing the addition of two natural numbers 
    (given in Peano notation via zero $\O$ and the successor function $\ts$) 
    is defined by the two rewrite rules
    $\tplus(\ts(x),y) \to \ts(\tplus(x,y))$ and $\tplus(\O,y) \to y$.
    $\R$ allows for the reduction $\tplus(\ts(\O),\ts(\O)) \to_{\R} \ts(\tplus(\O,\ts(\O))) \to_{\R} \ts(\ts(\O))$.
    For derivational complexity analysis, we can consider the wARS 
    $(\TSet{\Sigma}{\VSet}, \mto_{\R}, \seminatext, 
    \fNF^{\mathsf{cplx}}, \mathsf{Aggr}_{a \mto B}^{\mathsf{cplx}})$.
\end{example}

If the ARS $(A, \to)$ is finitely non-deterministic, 
then $(A, \to)$ is terminating if and only if
${\semelem{a}}_{\cplx{A, \to}} < \infty$ for every $a \in A$.
While the ``if'' direction holds for any ARS,
due to possibly infinite non-determinism, 
the ``only if'' direction does not hold in general.

\begin{example}[Non-Deterministic ARS]\label{ex:termination-analysis}
    Consider the ARS $(\IN_a, \to)$ with $\IN_a = \{a\} \cup \IN$ and 
    $\to \; = \{a \to n \mid n \in \IN\} \cup \{n+1 \to n \mid n \in \IN\}$ from
    \cite{avanzini2020probabilistic}. 
    For $\cplx{\IN_a, \to}$, we have $\semelem{a} = \infty$ as for\linebreak[3] all $n \in \IN$
    there is a $(\IN_a, \to)$-RT of depth $n+1$ with root $a$.
    However, $(\IN_a, \to)$ is terminating.
\end{example}

The definition of $\cplx{A,\to}$ can also be adjusted to prove termination
and analyze derivational complexity of \emph{sequence} ARSs.

\subsection{Size Bounds}\label{Expressability:Size}

In addition to the runtime of a program, its memory footprint is of interest as well.
Consider an operating system which should be able to run forever.
However, during this infinite execution, certain values that are stored in memory
must not become arbitrarily large, i.e., no overflow should occur. 
To analyze this, we can use the arctic semiring $\semiarc$.

\begin{example}[Memory Consumption of Operating System]\label{ex:memory-analysis}
  Consider a very simplified operating system\footnote{See \Cref{WP Comparison 2} for a
  more involved operating system algorithm that guarantees
  mutual exclusion.}
    with two processes $P_1$ and $P_2$ that should be performed repeatedly.
    The operating system can either be idle, run a process, 
    or add a process at the end of the waiting queue.
    We represent this by the ARS
    $(\mathsf{OS},\to)$ with 
    $\mathsf{OS} = \{\tidle(p), \twait(p), \trun(p) \mid p \in \{P_1, P_2\}^*\}$. 
    So an object from $\mathsf{OS}$ represents the current state of the
    operating system ($\tidle$, $\twait$, or $\trun$) and the
    current waiting queue $p$. 
    The rules of the ARS are
    $\tidle(p) \to \twait(p)$, $\tidle(p) \to \trun(p)$ 
    (add a new process to the waiting queue or run some process), 
    $\twait(p) \to \tidle(p P_1), \twait(p) \to \tidle(p P_2)$ 
    (add $P_1$ or $P_2$ to the waiting queue),
    and $\trun(P_1 p) \to \tidle(p), \trun(P_2 p) \to \tidle(p)$ 
    (run the process waiting the longest) for all $p \in \{P_1, P_2\}^*$.
    We use the wARS $(\mathsf{OS}, \to, \semiarc, 
    \fNF^{\mathsf{size}}, \mathsf{Aggr}_{a \to B}^{\mathsf{size}})$ 
    with $\fNF^{\mathsf{size}}(\trun(\varepsilon)) = 0$
    for $\NF_\to = \{ \trun(\varepsilon) \}$
    and
    \[\begin{array}{ccccccc}
        \mathsf{Aggr}_{\tidle(p) \to \twait(p)}^{\mathsf{size}} &=& \mathsf{Aggr}_{\tidle(p) \to \trun(p)}^{\mathsf{size}} &=& v_1 \\
        \mathsf{Aggr}_{\trun(P_1 p) \to \tidle(p)}^{\mathsf{size}} &=& \mathsf{Aggr}_{\trun(P_2 p) \to \tidle(p)}^{\mathsf{size}} &=& v_1 \\
        \mathsf{Aggr}_{\twait(p) \to \tidle(p P_1)}^{\mathsf{size}} &=& \mathsf{Aggr}_{\twait(p) \to \tidle(p P_2)}^{\mathsf{size}} &=& v_1 \semiplus_{\semiarc} \left( |p| + 1 \right) &=& \max\{v_1,|p| + 1\}.\!
    \end{array}\]

    \noindent
    Note that we may have a different aggregator for every
    sequence $p \in \{P_1, P_2\}^*$, i.e., $|p| + 1$ is a constant.
    We obtain $\semelem{\tidle(\varepsilon)} = \infty$, proving that a
    reduction leading to a waiting queue of unbounded size exists.
\end{example}

\subsection{Probabilistic Rewriting}\label{Expressability:Probabilistic}

In \cite{avanzini2020probabilistic,BournezRTA02,bournez2005proving}, ARSs were extended to the probabilistic setting. 
The relation $\tored{}{}{}$ of a probabi\-listic ARS
has (countable) multi-distributions on the right-hand sides.
A \emph{multi-distribution} $\mu$\linebreak[3] on a set $A \neq \emptyset$ is a countable multiset
of pairs $(p:a)$, where $p \in \RR$ with
$0 < p \leq 1$ is a probability and $a \in A$, 
with $\sum_{(p:a) \in \mu} \, p = 1$.  
$\Dist(A)$ is the set of all multi-distributions on $A$ and
$(A, \tored{}{}{})$ with 
$\tored{}{}{} \;\subseteq A \times \Dist(A)$ 
is a \emph{probabilistic abstract reduction system} (pARS). 
Depending on the property of interest, we can, e.g., use
the semiring
$\semirealext$ to describe the termination probability 
or even the expected derivational complexity of the pARS.

\begin{example}[Random Walk]\label{ex:probabilistic-analysis}
    Consider the biased random walk on $\IN$ given by the proba\-bilistic relation
    $\tored{}{}{} \; = \{n + 1 \tored{}{}{} \{\nicefrac{2}{3}:n, \nicefrac{1}{3}:n+2\} \mid n \in \IN\}$.
    We use the sARS $(\IN,\to)$ with $\to \; = \{n + 1 \to [n, n+2] \mid n \in \IN\}$,
    the semiring $\semirealext$,
    the interpretation of the normal form $\fNF(0) = 1$, 
    and the aggregator $\mathsf{Aggr}_{n+1 \to [n,n+2]}
    = \left( \nicefrac{2}{3} \semimult_{\semirealext} v_1 \right) \semiplus_{\semirealext} \left( \nicefrac{1}{3} \semimult_{\semirealext} v_2 \right) 
    = \nicefrac{2}{3} \cdot v_1 + \nicefrac{1}{3} \cdot v_2$ for every $n \in \IN$.
    The weight of the tree $\FT$ from \Cref{fig:red-tree-2} 
    is $\semantics{\FT}{}{} = \nicefrac{4}{9}$, since $\nicefrac{4}{9}$ is
    the probability to reach $0$ within two steps.
    The weight of the infinite extension $\FT_{\infty}$ of the depicted tree
    $\FT$ is $\semantics{\FT_{\infty}}{}{} = 1$,
    as such a random walk terminates with probability $1$. In this way
    one can use semiring semantics to express 
    \emph{almost-sure termination} (AST) of pARSs \cite{avanzini2020probabilistic}.

    Obviously, we can also consider infinite-support distributions, e.g.,
    consider the probabilis\-tic relation
    $\tored{}{}{} \; = \{n + 1 \tored{}{}{} \{\Geo(m):m \mid m \in \IN\} \mid n \in \IN\}$,
    where $\Geo$ denotes the \emph{geometric distribution}, i.e., $\Geo(m) = (\nicefrac{1}{2})^{m+1}$
    for all $m \in \IN$.
    Here, we use the sequence ARS $(\IN, \to)$ with $\to \; = \{n + 1 \to [0, 1, 2, \ldots] \mid n \in \IN\}$,
    and the aggregator $\mathsf{Aggr}_{n+1 \to [0, 1, 2, \ldots]} 
    = \semiplusbig_{m = 0}^{\infty} \left( \Geo(m) \semimult_{\semirealext} v_{m+1}
    \right)$ for every $n \in \IN$ ($\fNF$ and $\semirealext$ remain as
    above).

    Moreover, we can also use different aggregators and interpretations of normal forms
    to analyze the probability of reaching a certain normal form, 
    or the expected complexity (i.e., the expected number of reduction steps).
    For the expected derivational complexity of the biased random walk, 
    we again use the semiring $\semirealext$ 
    but switch to the interpretation of the normal form $\fNF(0) = 0$
    and the aggregator $\mathsf{Aggr}_{n+1 \to [n,n+2]}
    = 1 + \nicefrac{2}{3} \cdot v_1 + \nicefrac{1}{3} \cdot v_2$, 
    i.e., we add $1$ in each step and start with $0$.
    Then we obtain $\semelem{n} \neq \infty$ for every $n\in \IN$,
    i.e., the expected derivational complexity is finite for each possible start of the random
    walk, which proves \emph{positive} and \emph{strong almost-sure termination} (PAST and SAST)
    \cite{avanzini2020probabilistic,bournez2005proving}.
\end{example}

\subsection{Formal Languages}\label{Expressability:Fairness}

We can use semirings like $\semilang$ to analyze the behavior of systems.
Reconsider the setting from \Cref{ex:memory-analysis}. 
Instead of analyzing the memory consumption of the waiting queue,
we can also analyze the possible orders of running processes.

\begin{example}[Process Order for Operating System]\label{ex:fairness-analysis}
  Reconsider the 
  sARS for the operating system from \Cref{ex:memory-analysis}. 
    We can use the wARS $(\mathsf{OS}, \to, \semilang, 
    \fNF^{\mathsf{fair}}, \mathsf{Aggr}_{a \to B}^{\mathsf{fair}})$ with $\Sigma = \{P_1, P_2\}$, 
    $\fNF^{\mathsf{fair}}(\trun(\varepsilon)) = \semione_{\semilang} = \{\varepsilon\}$ and
  \[\begin{array}{ccccc}
        \mathsf{Aggr}_{\tidle(p) \to \twait(p)}^{\mathsf{fair}} &=& \mathsf{Aggr}_{\tidle(p) \to \trun(p)}^{\mathsf{fair}} &=& v_1 \\
        \mathsf{Aggr}_{\twait(p) \to \tidle(p P_1)}^{\mathsf{fair}} &=& \mathsf{Aggr}_{\twait(p) \to \tidle(p P_2)}^{\mathsf{fair}} &=& v_1 \\
        &&\mathsf{Aggr}_{\trun(P_1 p) \to \tidle(p)}^{\mathsf{fair}} &=& \{P_1\}
        \semimult_{\semilang} v_1 \\
        &&\mathsf{Aggr}_{\trun(P_2 p) \to \tidle(p)}^{\mathsf{fair}} &=& \{P_2\} \semimult_{\semilang} v_1 \\
    \end{array}\]

    \noindent
    Our operating system allows running the processes in any order,
    since $\semelem{\tidle(\varepsilon)} = \Sigma^*$.
\end{example}

\subsection{Combinations of Semirings}\label{Expressability:WP}

Cartesian products (and even matrices) of semirings form a semiring again 
by performing addition and multiplication pointwise.
Moreover, if all the semirings are complete lattices, then so is the 
resulting Cartesian product semiring.

\begin{restatable}[Cartesian Product Semiring]{lemma}{TupleSem}\label{lem:tuple-semiring}
    Let $(\IS_i)_{1 \leq i \leq n}$ be a family of complete lattice semirings. 
    Then $\IS = \bigtimes_{i = 1}^{n} \IS_i$ is a complete lattice semiring with
    $(x_1, \ldots, x_n) \semiplus_{\IS} (y_1, \ldots, y_n)
    = (x_1 \semiplus_{\IS_1} y_1, \ldots, x_n \semiplus_{\IS_n} y_n)$,
    and $(x_1, \ldots, x_n) \semimult_{\IS} (y_1, \ldots, y_n)
    = (x_1 \semimult_{\IS_1} y_1, \ldots, x_n \semimult_{\IS_n} y_n)$.
\end{restatable}

\begin{example}[Analyzing Complexity and Safety Simultaneously]\label{ex:tuple-semi}
    Consider the ARS
    $(\IZ, \to)$ with $\to \;= \{n \to n-2 \mid n~\text{odd}\} 
    \cup \{n \to n+2 \mid n \leq -2, n~\text{even}\} 
    \cup \{n \to n-2 \mid n \geq 2, n~\text{even}\}$.
      Additionally, consider a certain ``unsafe'' property, e.g., hitting an even number.
    To analyze whether all infinite sequences are safe, 
    we take the product semiring
    $\IS = \seminatext \times \semibool$ 
    over
    $\seminatext$ and the \emph{Boolean semiring} $\semibool$.  
    We use the normal form interpretation $\fNF(0) = (0,\ttrue)$ and 
    the aggregator $\mathsf{Aggr}_{n \mto [m]} = (1,[n \! \mod 2 = 0]) \semiplus_\IS v_1$ 
    for every $n \to m$.
    The first component describes the derivational complexity, 
    while the second describes 
    whether we reached an even number at some point
    during the reduction.
    In \Cref{Proving Upper Bounds}, we will see how to prove boundedness (i.e.,
    $\semelem{n} \neq (\infty, \ttrue)$ for every $n \in \IZ$)
    indicating safety of every infinite reduction.
\end{example}

Taking tuples for verification is not the same as performing two separate analyses.
Analy\-zing safety and complexity on their own for the ARS from 
\Cref{ex:tuple-semi} would fail, 
since the ARS is neither safe nor has finite complexity for every $n \in \IZ$.
Note the change of quantifiers: Instead of
``\emph{all runs are safe, or all runs are finite}'', 
we prove ``\emph{all runs are finite or safe}''.

\subsection{Limitations}\label{Expressability:Limit}

The following example illustrates a limit 
of our approach.

\begin{example}[Starvation Freedom]\label{ex:starvation-freedom}
   To analyze \emph{starvation freedom}, i.e., whether every process will eventually be served,
    one can use the tuple semiring $\seminatext\!\times\!\seminatext$ 
    for our operating system from \Cref{ex:memory-analysis} (a corresponding more
    complex example for starvation freedom is presented in \Cref{WP Comparison 2}).
  Now the two entries of the tuples count how often a process was already served.
    Thus,  we can use the wARS $(\mathsf{OS}, \to, \seminatext\!\times\!\seminatext, 
    \fNF^{\mathsf{starv}}, \mathsf{Aggr}_{a \to B}^{\mathsf{starv}})$ with
    $\fNF^{\mathsf{starv}}(\trun(\varepsilon)) =
    \seminull_{(\seminatext\!\times\!\seminatext)} =
 (0,0)$ and the aggregator
    \[\begin{array}{ccccc}
        \mathsf{Aggr}_{\tidle(p) \to \twait(p)}^{\mathsf{starv}} &=& \mathsf{Aggr}_{\tidle(p) \to \trun(p)}^{\mathsf{starv}} &=& v_1 \\
        \mathsf{Aggr}_{\twait(p) \to \tidle(p P_1)}^{\mathsf{starv}} &=& \mathsf{Aggr}_{\twait(p) \to \tidle(p P_2)}^{\mathsf{starv}} &=& v_1 \\
        && \mathsf{Aggr}_{\trun(P_1 p) \to \tidle(p)}^{\mathsf{starv}}  &=& (1,0) \semiplus v_1 \\
        && \mathsf{Aggr}_{\trun(P_2 p) \to \tidle(p)}^{\mathsf{starv}} &=& (0,1) \semiplus v_1 \!
    \end{array}\]
However, starvation freedom cannot be analyzed via our current definition of $\semelem{a}$
in \Cref{def:element-semiring-semantics}.
We have
$\semelem{a} = (\infty, \infty)$ for all 
     $a \in \mathsf{OS} \setminus
\{\trun(\varepsilon)\}$ (i.e., for all non-normal forms).  This means that
for every such start
configuration $a$, there \emph{exists} a (``worst-case'') reduction of weight  
$(\infty, \infty)$ where both processes are served infinitely often.
However, for starvation freedom, one\linebreak[3] would have to show that \emph{every} infinite reduction
serves both processes infinitely often (i.e., this would need to hold
irrespective of 
 how  the non-determinism in the reductions is resolved).
However, $(\mathsf{OS}, \to)$ is not starvation free, 
    since, e.g.,  we may only serve $P_1$ infinitely often.

    So a property like starvation freedom cannot be expressed with our current definition
of $\semelem{a}$, because due to the use of the least upper bound in \Cref{def:element-semiring-semantics}, here we
only focus on worst-case reductions. An extension of our approach to also analyze (bounds
on) best-case reductions in order to prove properties like starvation freedom is  an interesting direction for future work.
\end{example}


\section{Proving  Upper Bounds on Weights}\label{Proving Upper Bounds}

In this section, we present a technique amenable to automation which aims to prove an
upper bound on all possible weights
${\semelem{a}}$ of objects $a \in A$, i.e., it shows that
${\semelem{a}} \neq \top$ for all  $a \in A$.
For the remaining sections, we fix a wARS $\wARS$.

\begin{definition}[Boundedness]
  A wARS is \defemph{bounded} if
  ${\semelem{a}} \neq \top$ for all  $a \in A$.
\end{definition}

The examples in \Cref{Expressability} illustrate that boundedness is a crucial
property for wARSs, and that depending on the semiring, on the interpretation of
the normal forms, and on the aggregators, boundedness may have completely
different implications.

We first establish sufficient conditions for boundedness of a wARS in \Cref
{Guaranteed Boundedness}. Afterwards, we show in \Cref{Interpretation Method}
that the well-known interpretation method can be generalized to prove
boundedness for wARSs where these conditions are not satisfied.

\subsection{Guaranteed Boundedness}\label{Guaranteed Boundedness}

One can directly guarantee boundedness by an adequate choice of the semiring,
the inter\-pretation of the normal forms, and the aggregators. We say
that $\fNF: \NFto \to \IS$ is \emph{universally bounded} if
there exists a universal bound $\bound \in \IS\setminus \{\top\}$ with $\fNF
(a) \semileq
\bound$ for all $a\in \NFto$.
An aggregator function $\aggrrule : \IS^{|B|} \to \IS$ 
is \emph{selective} if for every $[s_1,s_2, \ldots] \in \IS^{|B|}$ there exists an
$1 \leq i \leq |B|$ such that $\aggrrule [s_1,s_2, \ldots] = s_i$.
For example, 
in the \emph{bottleneck semiring} $\semibottle = \semibottlelong$, finite
aggregator functions without constants are always
selective, since $\max$ and $\min$ are selective functions.

\begin{restatable}[Sufficient Condition for Boundedness (1)]{theorem}{SuffBoundOne}\label{thm:guaranteed-bounded1}
    A wARS is
    \begin{itemize}
        \item not bounded if $\fNF(a) = \top$ for some
        $a\in \NFto$.

        \item bounded if $\fNF$ is universally bounded and all
        $\aggrrule$ are selective.
    \end{itemize}
\end{restatable}

\medskip

Next, we do not only consider properties of $\fNF$ and
$\aggrrule$, but also properties of the sARS in order to
guarantee boundedness.

\begin{example}[Boundedness for Provenance Analysis Example]\label{ex:guaranteed-boundedness}
    Reconsider the setting of \Cref{ex:boolean-formula} and the sARS
    of \Cref{ex:boolean-formula-as-mARS}. Note that
    all propositional formulas are finite, hence the sARS is
    finitely branching and terminating. If none of the atomic facts has infinite cost,
    then no formula has infinite cost, since finite sums and
    products in the arctic semiring $\semiarc = \semiarclong$ cannot result in
    $\infty$ if all of its arguments are smaller than $\infty$.
\end{example}

The latter property of the semiring 
is called the \emph{extremal property} 
(or \emph{convex hull concept}).

\begin{definition}[Extremal Property]\label{def:extremal}
    A function $f:\IS^n\to \IS$ over a semiring $\IS$ with $n\in \IN$
    has the \defemph{extremal property} if
    $f(e_1, \ldots, e_n) \neq \top$ for all $e_1, \ldots, e_n\in \IS
    \setminus\{\top\}$.
    A semiring $\IS = \semilong$ has the extremal property if 
    $\semiplus$ and $\semimult$ have the extremal property.
\end{definition}

If addition and multiplication of a semiring $\IS$ satisfy the extremal
property, then sums and products 
of finite sequences $T \subseteq \IS \setminus \{\top\}$ do not
evaluate to $\top$, i.e., $ \semiplusbig T \neq \top \text{ and } \semimultbig
T \neq \top$.
Thus, every finite aggregator function that does
not use the constant $\top$ never evaluates to $\top$.
However, this does not necessarily hold for infinite
sums, products, and aggregators.
Consider, e.g., the subset $\IN \subset \IN^{\infty}$ of the extended natural
numbers, where $\semiplusbig_{\IN^{\infty}} \IN = \sum \IN = \infty = \top_{\IN^{\infty}}$.
Selective functions always satisfy the extremal property. 

\begin{example}[Extremal Property]
    \Cref{ex:guaranteed-boundedness} shows that the arctic semiring $\semiarc$
    has the extremal property. The extended naturals semiring
    $\seminatext$ also has the extremal property, since $a +
    b \neq \infty$ and $a \cdot b \neq \infty$ for all $a, b\in\IN$. Actually,
    the extremal property holds for all semirings in \cref{table:semirings} except
    for the
    formal languages semiring $\semilang$, where, e.g., $\big
    (\Sigma^* \setminus \{\varepsilon\}\big) 
    \cup \{\varepsilon\} = \Sigma^* = \semitop_{\semilang}$. 
    Cartesian products of semirings with the extremal property do not
    necessarily satisfy the extremal property again:
    Consider $\semiarc \times \semiarc$,
    where the addition of two objects that are different from $\semitop_{\semiarc \times \semiarc}$
    yields $(\seminull,\semitop) \semiplus
    (\semitop, \seminull) = (\semitop, \semitop) = \semitop$.
\end{example}

This yields another sufficient condition for boundedness
(see \Cref{ex:guaranteed-boundedness}).

\begin{restatable}[Sufficient Condition for Boundedness (2)]{theorem}{SuffBoundTwo}\label{thm:guaranteed-bounded}
    A wARS is bounded if
   \begin{itemize}
        \item the sARS $(A, \to)$ is terminating, finitely
        non-deterministic, and finitely branching,

      \item the semiring $\IS$ has the extremal property,

        \item $\fNF(a) \neq \top$ for all $a \in \NFto$, and

        \item all aggregators $\aggrrule$ are finite and do not use $\top$ as a constant.
    \end{itemize}
\end{restatable}

\medskip

While the requirements in \Cref{thm:guaranteed-bounded} may seem restrictive,
the ones on $\fNF$ and $\aggrrule$ only consider certain edge cases.
However, termination of the underlying sARS is a crucial requirement
for \Cref{thm:guaranteed-bounded}
that one has to prove beforehand.

\cref{ex:termination-analysis} presents an unbounded wARS
which satisfies many of the
constraints from \cref{thm:guaranteed-bounded}, but illustrates the importance
of \emph{finite} non-determinism. Infinite non-determinism allows the existence
of a chain of reduction trees with ascending weights which reaches $\top$
in the limit.

\subsection{Proving Boundedness via Interpretations}\label{Interpretation Method}

To handle wARSs that neither satisfy the requirements of \Cref{thm:guaranteed-bounded1}
nor of
\Cref{thm:guaranteed-bounded}, 
we now extend the well-known interpretation method (see, e.g., \cite{lankford1979ProvingTermRewriting}) to prove
boundedness of general wARSs. 
In some cases, e.g., when considering term rewriting as in \Cref{Expressability:Runtime},
this often allows proving termination automatically.
Before presenting the technique to prove boundedness via interpretations, we
introduce the notions of monotonicity and continuity. 

\begin{definition}[Monotonicity, Continuity]\label{def:monotonic}
    A function $f: \IS \rightarrow \IS$ on a semiring $\IS$ is 
    \defemph{monotonic} if for all $s,t \in \IS$,
    $s \semileq t$ implies $f(s) \semileq f(t)$.
    It is \defemph{continuous} if for all $T\subseteq \IS$,
    $\semisup f(T) = \semisup \{f(t) \mid t \in T\} = f(\semisup T)$.
    A function $f: \IS^n \rightarrow \IS$ with $n \geq 2$ is monotonic (continuous) if it is
    monotonic (continuous) in every argument.
\end{definition}

The natural order implies monotonicity of $\semiplus$ and $\semimult$,
and thus, every aggregator function is monotonic as well.
An analogous result is obtained if additionally
$\semiplus$ and $\semimult$ are continuous.

\begin{restatable}[Monotonicity and Continuity of Aggregator Functions]{lemma}{MonoContAggregator}\label{lem:mono-con-agg}
    For a semiring, the operations $\semiplus$, $\semimult$,
    and all aggregator functions are monotonic.
    Moreover, if $\semiplus$ and $\semimult$ are continuous, then so are all aggregator functions.
\end{restatable}

Now we show how to use interpretations to prove boundedness of a wARS.
The idea is to\linebreak[3] use an ``embedding'' (or ``ranking function'')
$\mathfrak{e}$ which maps every object from $A$ to a non-maximal element of the semiring
$\IS$. Due to monotonicity of all aggregator functions,
the conditions of \Cref{thm:boundedness}
ensure that for all nodes $v$ in any
finite-depth reduction tree $\FT$,
we have $\mathfrak{e}(a_v) \semigeq \semantics{\FT}{v}{}$. 
Hence, $\mathfrak{e}(a)$ is a bound on $\semantics{\FT}{}{}$ 
for all reduction trees $\FT \in \Phi(a)$.

\begin{restatable}[Sufficient and Necessary Condition for Boundedness]{theorem}{BoundSoundAndComplete}\label{thm:boundedness}
    A wARS is bounded if
    there exists an \defemph{embedding}
    $\mathfrak{e}: A\to \IS\setminus\{\top\}$ such that 
    \begin{itemize}
        \item $\mathfrak{e}(a) \semigeq \fNF(a)$ for all
        $a \in \NFto$ and

        \item $\mathfrak{e}(a) \semigeq
            \aggrrule [ \mathfrak{e}(b) \mid b\in B ]
        $ for all $\rulesemi$.
    \end{itemize}
    Then $\mathfrak{e}(a) \semigeq \semelem{a}$ for all $a
    \in A$.
    The reverse (``only if'')
    holds if 
    $\semiplus$ and $\semimult$ are continuous.
\end{restatable}

As shown in \Cref{Expressability:Probabilistic}, for every probabilistic ARS,
we can obtain a corresponding wARS to analyze PAST/SAST. Hence, \Cref{thm:boundedness}
allows proving boundedness of this wARS which then implies PAST/SAST of the
original probabilistic ARS.

\begin{example}[Expected Runtime of Probabilistic ARSs]\label{PAST example}
    We use \cref{thm:boundedness} to prove that the expected
    derivational complexity of the biased random walk in \cref{ex:probabilistic-analysis}
    is finite.
    We use the embedding $\mathfrak{e}(n) = 3 \semimult n = 3\cdot n$ for all $n \in \IN$. 
    Note that we indeed have $\mathfrak{e}(n) \neq \top = \infty$ for all $n \in A = \IN$.
    Moreover, for $0$ (the only normal form) we have
    $\mathfrak{e}(0) = 0 = \fNF(0)$.
    Regarding the reduction steps, we have for any $n+1 \to [ n, n+2 ]$ with $n \in \IN$:
    \vspace*{-.15cm}
    \[\begin{array}{crcl}
        &\mathfrak{e}(n+1) &\semigeq& \aggr{n+1 \to [n, n+2]}[\mathfrak{e}(n), \mathfrak{e}(n+2)] \\
        \iff&3 \semimult (n+1) &\semigeq& \aggr{n+1 \to [n, n+2]}[3\semimult n, 3 \semimult (n+2)] \\
        \iff&3 + 3\cdot n  &\semigeq& 1 \semiplus \nicefrac{2}{3} \semimult 3\semimult n
        \semiplus \nicefrac{1}{3} \semimult 3 \semimult (n+2) \;=\; 3 + 3\cdot n \!
    \end{array}\]
    \vspace*{-.3cm}

    \noindent
    By \cref{thm:boundedness}, the wARS is bounded, which means
    that the expected derivational complexity of the biased random walk is finite for every
    starting position $n\in\IN$. The embedding $\mathfrak{e}$ gives us a bound on the
    expected complexity as well, i.e., by \cref{thm:boundedness} we infer that the expected number of
    steps is at most \emph{three} times the start position $n$.
\end{example}

\begin{example}[Termination of TRSs]\label{Termination Example}
    The next example shows how our approach can be used for automated termination proofs
    of term rewrite systems. 
    Reconsider the TRS $\R$ from \cref{ex:complexity-analysis}
    with the wARS $(\TSet
    {\Sigma}{\VSet}, \mto_{\R}, \seminatext,
    \mathsf{f}_{\NF}^{\mathsf{cplx}}, \mathsf{Aggr}_{a \mto B}^{\mathsf
    {cplx}})$.
    We define the embedding $\mathfrak{e}: \TSet{\Sigma}{\VSet} \to \IN^\infty$
    with $\mathfrak{e}(t) \neq \infty$ for all $t \in \TSet{\Sigma}{\VSet}$
    recursively as $\mathfrak{e}(\O) = 0$, $\mathfrak{e}(\ts(t)) = \mathfrak{e}(t) \semiplus 1$, and 
    $\mathfrak{e}(\tplus(t_1, t_2)) = 2 \semimult \mathfrak{e}(t_1) \semiplus \mathfrak{e}(t_2) \semiplus 1$.
    To prove termination of the rewrite system $\R$ for all terms,
    we show that the two inequations
    required by \cref{thm:boundedness} hold for all instantiated 
    rewrite rules.\footnote{In order to \emph{lift} the inequations from rules
    to reduction steps, one has to ensure that the embedding $\mathfrak{e}$
    is \emph{strictly monotonic}, see, e.g., \cite{baader_nipkow_1998,terese2003term}.}
      For all $t \in \NF_{\mto_{\R}}$ we have $\mathfrak{e}(t) \semigeq \fNF^{\mathsf{cplx}}(\O) = 0$.
    For the rule $\tplus(\ts(x),y) \mto [ \ts(\tplus(x,y)) ]$ and all $t_1, t_2 \in \TSet{\Sigma}{\VSet}$ we get
    \vspace*{-.15cm}
    \[\begin{array}{crcl}
        &\mathfrak{e}(\tplus(\ts(t_1),t_2)) &\semigeq& \aggr{a \mto [b]}(\mathfrak{e}(\ts(\tplus(t_1,t_2)))) \\
        \iff&2 \semimult \mathfrak{e}(t_1) \semiplus \mathfrak{e}(t_2) \semiplus 3 &\semigeq& 1 \semiplus \mathfrak{e}(\ts(\tplus(t_1,t_2))) \\
        \iff&2 \cdot \mathfrak{e}(t_1) + \mathfrak{e}(t_2) + 3 &\semigeq& 2 \semimult \mathfrak{e}(t_1) \semiplus \mathfrak{e}(t_2) \semiplus 3 \;=\; 2 \cdot \mathfrak{e}(t_1) + \mathfrak{e}(t_2) + 3\!
    \end{array}\]
    \vspace*{-.3cm}

    \noindent
    and for the rule $\tplus(\O,y) \mto [y]$ we get $\mathfrak{e}(\tplus(\O, t_1)) \semigeq \aggr{a \mto [b]}(\mathfrak{e}(t_1)) \iff \mathfrak{e}(t_1) + 1 \semigeq \mathfrak{e}(t_1) + 1$.
    Again, by \cref{thm:boundedness} the wARS is bounded, hence the TRS terminates.
\end{example}

If one fixes the semiring $\IS$, 
the interpretation $\fNF$, and the aggregators $\aggr{a \to B}$,
searching for such an embedding $\mathfrak{e}$ 
can often be automated for arbitrary TRSs $\R$ using SMT solvers.

\begin{example}[Complexity and Safety]\label{ex:term-ComplexityandSafety}
    To prove boundedness of the wARS from \Cref{ex:tuple-semi},
    we use the embedding $\mathfrak{e}(n) = (\frac{|n|}{2},\ttrue)$ if $n \in \IZ$ is
    even
    and $\mathfrak{e}(n) = (\infty,\tfalse)$ if $n \in \IZ$ is odd.
    Then we have $\mathfrak{e}(0) = (0, \ttrue) = \fNF(0)$. For odd $n$, we obtain
     $\mathfrak{e}(n) = (\infty,\tfalse) = \mathsf{Aggr}_{n \mto [n-2]}(\mathfrak{e}(n-2))
    = \mathsf{Aggr}_{n \mto [n-2]}(\infty,\tfalse) = (1 + \infty, [n \mod 2 = 0] \vee
    \tfalse)$. For even $n \geq 2$, we get
    $\mathfrak{e}(n) = (\frac{n}{2}, \ttrue) = \mathsf{Aggr}_{n \mto
    [n-2]}(\mathfrak{e}(n-2)) = \mathsf{Aggr}_{n \mto
    [n-2]}(\frac{n}{2}-1, \ttrue) = (1 + \frac{n}{2} -1, [n \mod 2 = 0] \vee
    \ttrue)$. For  even $n \leq -2$, the reasoning is analogous.    
\end{example}


\section{Proving Lower Bounds on Weights}\label{Proving Lower Bounds}

Next, we discuss how to analyze lower bounds.
In \Cref{Approximating}, we show how to compute a
lower bound on weights $\semelem{a}$ and in \Cref{Increasing Loops}, 
we show how to prove unboundedness (i.e., $\semelem{a} = \top$).
Such lower bounds are useful to find bugs or potential attacks, e.g., 
inputs leading to very high computational costs in terms of
runtime or memory consumption.

\subsection{Approximating the Weight}\label{Approximating}

Since the weight $\semelem{a}$ is defined via the supremum of a set, we can approximate
$\semelem{a}$ from below by considering only nodes up to a certain depth
and only certain schedulers. In the following, for every reduction tree $\FT$ and $n \geq
0$, let $\FT|_n$ denote the tree that results from $\FT$ by removing all nodes with depth
$> n$.

\begin{corollary}[Lower Bound on Weight] \label{cor:approx}
    For any $a\in A$ and  $\Psi \subseteq \Theta \subseteq \Phi(a)$, we have
    $$
         \seminull \semileq \semisup\, \{ \semantics{\FT|_{n'}}{}{} \mid \FT\in \Psi, n' \leq n \}
        \semileq \semisup\, \{ \semantics{\FT|_{m'}}{}{} \mid \FT\in \Theta, m' \leq m \}
        \semileq \semelem{a}.
    $$
\end{corollary}

\Cref{cor:approx} states that we can approximate $\semelem{a}$ by starting with
some reduction tree $\FT$ with root $a$ and considering $\semantics{\FT|_0}{}{}$ as a
first approximation of the weight (where $\semantics{\FT|_0}{}{} = \seminull$ if 
$a \notin \NFto$). We can
consider nodes of larger depth ($\seminull \semileq \semantics{\FT|_0}{}{}
\semileq \semisup \{\semantics{\FT|_0}{}{}, \semantics{\FT|_1}{}{}\} \semileq
\semisup \{\semantics{\FT|_0}{}{}, \semantics{\FT|_1}{}{}, \semantics{\FT|_2}{}{}\}\semileq
\cdots \semileq \semelem{a}$) 
and more schedulers ($\seminull \semileq \semantics{\FT|_n}{}{} \semileq
\semisup \{\semantics{\FT|_n}{}{}, \semantics{\FT'|_n}{}{}\} \semileq \cdots \semileq \semelem{a}$) to refine this approximation.
However, we have to compute  $\semantics{\FT|_n}{}{}$ for all
reduction trees $\FT$ whose root is
labeled with $a$, 
leading to an exponential number of trees depending on the considered depth and the maximal 
number of non-deterministic choices between reduction steps of an object.

\Cref{thm:approx} shows that this number of calculations is not as high as it seems, 
and even feasible for deterministic sARSs.
Monotonicity of the aggregator functions (\Cref{lem:mono-con-agg}) ensures that we have $\seminull \semileq \semantics{\FT|_0}{}{} \semileq \semantics{\FT|_1}{}{} \semileq\semantics{\FT|_2}{}{} \semileq \cdots \semileq \semelem{a}$, i.e., 
to approximate the weight  up to depth $n \in \IN$, 
we do not have to compute $\semantics{\FT|_{n'}}{}{}$ 
for all $n' \leq n$, but just $\semantics{\FT|_n}{}{}$.
Moreover, we do not need to consider several 
reduction trees for deterministic systems, but just the supremum obtained when evaluating
the ``only possible'' reduction
tree ``as much as possible''.

\begin{restatable}[Approximating Deterministic Systems]{theorem}{Approximating}\label{thm:approx}
    Let $(A,\to)$ be a deterministic sARS. Then for every $a \in A$,
    there exists an $(A,\to)$-reduction tree $\FT$ whose root is labeled with $a$ such that
 $\seminull \semileq \semantics{\FT|_0}{}{} \semileq \semantics{\FT|_1}{}{} \semileq
    \semantics{\FT|_2}{}{} \semileq \cdots \semileq \semelem{a}$ and $\semisup \{\semantics{\FT|_n}{}{} \mid n \in \IN\} = \semelem{a}$.
\end{restatable}

\subsection{Proving Unboundedness by Increasing Loops}\label{Increasing Loops}

While the interpretation method of \Cref{thm:boundedness} 
is based on a technique to prove termination of ordinary ARSs,
finding loops (i.e., a non-empty reduction sequence $a \to \cdots \to a$)
is one of the basic methods to disprove termination.
If we find a finite-depth RT $\F{T}$ whose root is labeled with $a$ and some other node of
$\FT$ is also labeled with $a$,
then this obviously shows non-termination of the underlying sARS. The reason is that
we can obtain an RT of infinite depth by simply using the reduction steps from $a$ to
$a$ repeatedly.
For unboundedness, we additionally require that the weight increases with each loop iteration.
However, increasing weights are not sufficient for unboundedness, as shown
by the following example.

\begin{example}[Bounded with Increasing Loop]\label{ex:loop-formal-language}
    Let $A = \{a, b\}$ with $a \to [a]$ and $a \to [b]$.
    Moreover, we take the formal languages semiring $\semilang$ over $\Sigma = \{0,1\}$, the interpretation $\mathsf{f}_{\NF}(b) = \semione_{\semilang} = \{\varepsilon\}$,
    and the aggregators
    $\aggr{a \to [a]}(x) = \left( \{1\} \semimult_{\semilang}
    x \right) \semiplus_{\semilang} x = \{1w \mid w \in x\} \cup x$ and $\aggr{a \to [b]}(x)
    = x$. 
    Obviously, the wARS $(A, \to)$ admits a loop from $a$ to $a$.
    However, we have $\semelem{a} = \{1\}^* \neq \Sigma^* = \semitop$, 
    even though we have a loop with increasing weights.
\end{example}

The problem in \Cref{ex:loop-formal-language} is that $\top$ 
is not the least upper bound of the weights of the increasing loops.
Thus, to infer unboundedness, we require a fixed increase by an element $t$ in each iteration 
such that $\semiplusbig_{i = 1}^{\infty} t = \semitop$.
Then, we can use the least upper bound of the series of increasing loops
as a lower bound for $\semelem{a}$, showing unboundedness.

Before we present the corresponding theorem, we define a partial 
evaluation of finite-depth trees where the label of the leaf $v_0$ is replaced by a variable $X$.

\begin{definition}[Induced Weight Polynomial]\label{Induced Weight Polynomial}
  Let $\FT$ be a  RT of finite depth with a leaf $v_0$ and let $X$ be a specific variable. Then we define:
    \begin{align*}
        \semantics{\FT}{v_0}{v_0} &= X && \\
        \semantics{\FT}{v}{v_0} &= \fNF(a_v) \; &&\text{if $v \neq v_0$ and $a_v \in
          \NFto$},\\
 \semantics{\FT}{v}{v_0} &= \seminull \; &&\text{if $v \neq v_0$ is a leaf and $a_v \notin
   \NFto$},\\
     \semantics{\FT}{v}{v_0} &= \aggr{a_v\to B}\left[
            \semantics{\FT}{w}{v_0} \mid w \in vE \right]
            &&\text{if $v$ is an inner node and $B = [a_w \mid w \in vE]$}
    \end{align*}
    The \defemph{induced weight polynomial} $\PP_{v_0}(\F{T}) \in \IS[X]$ is defined by
    $\PP_{v_0}(\F{T}) = \semantics{\FT}{r}{v_0}$, where $r \in V$ is the root of
    $\F{T}$.
\end{definition}

\begin{example}[Induced Weight Polynomial]\label{ex:induced-weight-poly}
    Reconsider the ARS $(\mathsf{OS}, \to)$ from \Cref{ex:memory-analysis} in \Cref{Expressability:Size}.
    We have the loop 
    $\tidle(\varepsilon) \to \twait(\varepsilon) \to \tidle(P_1) \to \trun(P_1) \to \tidle(\varepsilon)$
    and the corres\-pon\-ding 
    finite reduction tree can be seen in \Cref{fig:reduction_tree_loop}. Here, the only
    leaf $v_0$ is labeled by $\tidle(\varepsilon)$.\linebreak[3]
    If we consider the runtime by using the wARS $\cplx{\mathsf{OS}, \to}$,
    then we get the induced weight\linebreak[3] polynomial $X+4$.
    If we consider the size of the waiting list instead, i.e., the wARS $(\mathsf{OS}, \to,\linebreak[3] \semiarc, 
    \fNF^{\mathsf{size}}, \mathsf{Aggr}_{a \to B}^{\mathsf{size}})$ from \Cref{Expressability:Size},
    then we get the induced weight polynomial $\max\{X,1\}$.
\end{example}

\begin{restatable}[Proving Unboundedness via Increasing Loops]{theorem}{Loops}\label{thm:loops}
    Let $\F{T}$ be a finite RT 
    where both the root $r$ and a leaf $v_0 \neq r$ are labeled with $a$. 
    Moreover, let $t \in \IS$ with $\semiplusbig_{i = 1}^{\infty} t = \semitop$.
    If $\PP_{v_0}(\F{T})(s) \semigeq s \semiplus t$ for all 
    $s \in \IS$,
    then $\semelem{a} = \top$.
\end{restatable}

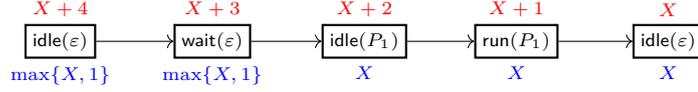
\begin{figure}
    \centering
    \scriptsize
    \begin{tikzpicture}
        \tikzstyle{myRect}=[thick,draw=black!100,fill=white!100,minimum size=4mm, shape=rectangle]
        \tikzstyle{trueRect}=[thick,draw=black!100,fill=lipicsYellow!100,minimum size=4mm, shape=rectangle]

        \node[myRect, label=below:{$\textcolor{blue}{\max\{X,1\}}$}, 
        label=above:{$\textcolor{red}{X+4}$}] at (0, 0)
            (a) {$\tidle(\varepsilon)$};

        \node[myRect, label=below:{$\textcolor{blue}{\max\{X,1\}}$}, 
        label=above:{$\textcolor{red}{X+3}$}] at (2, 0)
            (b1) {$\twait(\varepsilon)$};

        \node[myRect, label=below:{$\textcolor{blue}{X}$}, 
        label=above:{$\textcolor{red}{X+2}$}] at (4, 0)
            (c1) {$\tidle(P_1)$};
        \node[myRect, label=below:{$\textcolor{blue}{X}$}, 
        label=above:{$\textcolor{red}{X+1}$}] at (6, 0)
            (d1) {$\trun(P_1)$};
        \node[myRect, label=below:{$\textcolor{blue}{X}$}, 
        label=above:{$\textcolor{red}{X}$}] at (8, 0)
            (e) {$\tidle(\varepsilon)$};

        \draw (a) edge[->] (b1);
        \draw (b1) edge[->] (c1);
        \draw (c1) edge[->] (d1);
        \draw (d1) edge[->] (e);
    \end{tikzpicture}
    \caption{An example of a finite reduction tree containing a loop from $\tidle(\varepsilon)$ to itself.
    The corresponding evaluation of the induced weight polynomial for the runtime is
    depicted in red above the nodes, and for the space of the waiting list in blue below the nodes.}
    \label{fig:reduction_tree_loop}
\end{figure}

\begin{example}[Unbounded Runtime]\label{Unbounded Example}
    Continuing \Cref{ex:induced-weight-poly},
    we can use \Cref{thm:loops} to prove that the runtime of $(\mathsf{OS}, \to)$
    (i.e.,  the wARS $\cplx{\mathsf{OS}, \to}$) is unbounded.
    Since the induced weight polynomial of the loop $\FT$
    from \Cref{ex:induced-weight-poly}
    is $X + 4$,
    $\sum_{i = 1}^{\infty} 4 = \infty$, and 
    $\PP_{v_0}(\F{T})(s) \semigeq s + 4$ for every $s \in \seminatext$,
    we obtain $\semelem{\tidle(\varepsilon)} = \infty$.
    However, we cannot use \Cref{thm:loops} to prove that the
     memory consumption of $(\mathsf{OS}, \to)$ (i.e.,  the wARS $(\mathsf{OS}, \to, \semiarc, 
    \fNF^{\mathsf{size}}, \mathsf{Aggr}_{a \to B}^{\mathsf{size}})$ to express the size of the
    waiting list) is unbounded, as there is no
    $t\in \semiarc\setminus \{\top\}$ with
    $\semiplusbig_{i = 1}^{\infty} t = \semitop$.
\end{example}

There exist several automatic approaches to find loops in, e.g., term rewriting.
To lift these techniques to automatic unboundedness proofs for  wARSs
based on TRSs, one has to formalize the additional property 
$\exists t \in \IS: \semiplusbig_{i = 1}^{\infty} t = \semitop 
\land \forall s \in \IS: \PP_{v_0}(\F{T})(s) \semigeq s \semiplus t$
as an SMT problem over the corresponding semiring theory.


\section{Conclusion}\label{Conclusion}

We have developed semiring semantics for
abstract reduction systems using arbitrary complete lattice semirings.
These semantics capture and generalize numerous formalisms that have been studied in the literature.
Due to our generalization of these formalisms, 
we can now use techniques and ideas from, e.g., termination analysis,
to prove boundedness or other properties (or combinations of properties)
of reduction systems 
using a completely different semiring (e.g., a tuple semiring).
In the future, this may be used to improve the
automation of specific analyses and lead to
further applications of our uniform framework.

There are many directions for future work, e.g.,
one can try to improve our techniques
on proving and disproving boundedness.
In order to develop techniques amenable to automation
one could focus on term rewrite systems
where the reduction relation is represented by a finite set of rules.
For proving termination of TRSs, there exist more powerful techniques than just using
interpretations (e.g., the \emph{dependency pair framework} \cite{giesl2006mechanizing}).
Thus, we aim to develop a similar framework to analyze
boundedness for weighted TRSs in the future.
Currently our approach only focuses on (bounds on) worst-case reductions. Therefore, in
the future we will also investigate extensions in order to also express and analyze properties
like starvation freedom where one has to consider all (infinite) reductions.
We are also interested in adapting concepts like confluence and unique normal forms to weighted
rewriting, e.g., by studying rewrite systems where every
reduction tree that starts with the same object has the same weight if it
is evaluated ``as much as possible''.

\medskip

\noindent
\textbf{Acknowledgements:} We thank the reviewers for their useful remarks and suggestions.

\bibliography{biblio.bib}

\clearpage
\appendix

\section{Comparison to Weakest Preweightings for Weighted Imperative Programs}\label{WP Comparison}

In this appendix, we briefly compare our semantics to the formalism 
of \emph{weakest preweightings} ($\WP$) for weighted imperative programs from
\cite{BatzGKKW22}, and we present a more complex 
mutual exclusion algorithm from \cite{BatzGKKW22} which would
require
\emph{weakest liberal preweightings} ($\WLP$).

\subsection{Weakest Preweightings}\label{WP Comparison 1}
Assuming familiarity with the concepts and notation introduced in \cite{BatzGKKW22}, 
we show how to adapt them to our setting.
This should serve as a high-level illustration of the relationship between
weighted abstract reduction systems 
and the weighted imperative programs from \cite{BatzGKKW22}.
To do so, we show how to transform the imperative program 
modeling the ski rental problem from \cite{BatzGKKW22} 
into a weighted ARS using the tropical semiring $\semitrop$ such that
$\semantics{\delta_0}{}{}$ is equal to the 
``weakest preweighting''  $\WPSem{\mathsf{SkiAlg}}{\mathbf{1}}$
of the \emph{postweighting} $\mathbf{1}$ on the initial \emph{program configuration} $\delta_0$,
where $\mathsf{SkiAlg}$ is depicted in \Cref{alg1}.

\begin{wrapfigure}[10]{r}{0.22\textwidth}
    \begin{minipage}{0.22\textwidth}
        \vspace*{-.7cm}
        \begin{algorithm}[H]
            \DontPrintSemicolon%
            \SetInd{-.3cm}{.3cm}%
            \DecMargin{1cm}%
            \caption{}
            \label{alg1}
            \hspace*{-.5cm}   \While{$n > 0$}
            {
                $\{$\;
                $\phantom{\{} \!\semimult 1;$\;
                $\phantom{\{} n := n-1;$ \hspace*{-3.8cm}\; 
                $\} \semiplus \{$\;
                $\phantom{\{} \!\semimult y;$\;
                $\phantom{\{} n := 0;$\hspace*{-3.8cm}\; 
                $\}$
            }
        \end{algorithm}
    \end{minipage}
\end{wrapfigure}
The ski renting problem is a classical optimization problem.
Consider a person that does not own a pair of skis 
but is going on a skiing trip for an initially unknown number of $n \in \IN$ days.
At the dawn of each day, the person can decide whether to
rent a pair of skis for \EUR{1} for that day
or buy their own pair of skis for \EUR{$y$} and go skiing 
without any further costs, instead. 
What is the optimal strategy
for the person to spend a minimal amount of money?
In \cite{BatzGKKW22} this problem has been modeled 
via the imperative program $\mathsf{SkiAlg}$ depicted in \Cref{alg1}.
In general, weighted algorithms contain the standard control-flow instructions 
from the guarded command language ($\mathsf{GCL}$) 
syntax including non-deterministic branching ($\semiplus$),
with an additional $\semimult s$ statement,
where $s \in \IS$ is an element from a semiring.
The statement represents a $\mathsf{skip}$ operation (or $\mathsf{noop}$, 
i.e., skipping the execution step and doing nothing) that is used to give
a \emph{weight} to every execution path through the program.
The set of all such programs of the \emph{weighted guarded command language}
is denoted by $\mathsf{wGCL}$.

We first define the corresponding sARS $(A, \to)$ 
representing the possible computations of the program.
Here, we let $A$ be the set of all 
\emph{configurations} of the program
and $\to$ is defined via the \emph{transition} relation 
(\cite{BatzGKKW22}, Definition 3.1).
A configuration $(C, \sigma, n, w)$ consists of the remaining program 
$C \in \mathsf{wGCL} \cup \{\downarrow\}$,
where $\downarrow$ denotes an already terminated program;
an instantiation of the program variables $\sigma: \VSet \to \IN$, 
where $\Sigma$ is the set of all such instantiations;
a number of already performed execution steps $n \in \IN$;
and finally, a string $w \in \{L,R\}^*$ indicating the performed choices 
at non-deterministic steps.

\begin{definition}[sARS for Ski Renting Problem]\label{ski example 1}
    Following the notations from \cite{BatzGKKW22}, let $Q = (\mathsf{wGCL} \cup \{\downarrow\}) \times \Sigma \times \IN \times \{L,R\}^*$ be the set of all \defemph{configurations}.
    Moreover, let $\to \;= \{\delta \to [\delta_1, \delta_2, \ldots] \mid (\delta, w, \delta_i) \in \Delta\}$.
    Here, $\Delta \subseteq Q \times \IS \times Q$ denotes the transition relation according 
    to the small-step operational semantics given in \cite{BatzGKKW22},
    where $w$ is the weight of the transition $(\delta, w, \delta_i)$ from configuration
    $\delta$ to configuration $\delta_i$.
    To determine the order of the configurations $\delta_1, \delta_2, \ldots$, 
    we use an arbitrary total order on the transitions.
    The sARS representing the algorithm of the ski renting problem is $(Q,\to)$.
\end{definition}

To answer the question of the ski renting problem, 
one can compute the weakest preweighting $\WPSem{\mathsf{SkiAlg}}{\mathbf{1}}$
considering the tropical semiring $\semitrop = \semitroplong$
and the \emph{postweighting} $\mathbf{1}$.
When applying $\WPSem{\mathsf{SkiAlg}}{\mathbf{1}}$  to the initial configuration where
the program variable $n$ has the value $n_0$, one obtains 
$n_0 \semiplus y = \min\{n_0,y\}$, indicating that it is only beneficial to buy skis 
if they cost less than the number of skiing days.

A postweighting is a function $f: \Sigma \to \IS$ mapping final states 
to an element in the semiring, 
i.e., this corresponds to our interpretation of the normal forms $\fNF$.
The $\WP$ transformer aggregates the postweighting along the program execution paths by multiplication
and by summing over multiple possibilities during non-determinism.
In our setting, this can be expressed by choosing
aggregators that are weighted sums $\mathsf{Aggr}_{a \to B} = \sum_{1 \leq i \leq |B|} e_i \odot v_i$
for every $a \to B$ with $e_i \in \IS$ for all $1 \leq i \leq |B|$.

\begin{definition}[Aggregators and Interpretation for Ski Renting Problem]\label{ski example 2}
    Given a postweighting $f$, 
    we define the interpretation of normal forms
    as $\fNF = f$ and the aggregators as 
    $\mathsf{Aggr}_{a \to B} = \semiplusbig [w_i \semimult v_i \mid (a,w_i,b_i) \in \Delta]$.
\end{definition}

Thus, we obtain the following weighted abstract reduction system.

\begin{definition}[wARS for the Ski Renting Problem]
  The wARS for the ski renting problem is $(Q, \to, \semitrop, \fNF, \aggrrule)$,
with $Q$, $\to$ as in \Cref{ski example 1} and $\fNF$, $\aggrrule$ as in \Cref{ski example 2}.
\end{definition}

In the end, we get $\semantics{\delta_0}{}{} =
\WPSem{\mathsf{SkiAlg}}{\mathbf{1}} (\delta_0) = n_0 \semiplus y = \max\{n_0,y\}$,
where $\delta_0$ is the initial configuration 
that sets the program variable $n$ to the natural number $n_0 \in \IN$. 

Thus, we can express weakest preweightings from \cite{BatzGKKW22}
in our formalism.
One might think that we can even express more than with weakest preweightings, 
as we can use aggregators that 
are not just simple weighted sums but arbitrary polynomials.
However, we expect that one can transform any $\mathsf{wGCL}$ program $C$ 
into another $\mathsf{wGCL}$ program $C'$ that can represent more complex
aggregators of $C$ by ordinary weighted sums of $C'$.
This might be an interesting direction for future research.

\subsection{Weakest Liberal Preweightings}\label{WP Comparison 2}

In order to express ``weakest liberal preweightings'' ($\WLP$) 
from \cite{BatzGKKW22} in our setting,
we would have to extend our definition of  $\semelem{a}$ in
\Cref{def:semiringSemantics}.
This is similar to the problem of analyzing best-case
reductions described in \Cref{Expressability:WP}.

\begin{wrapfigure}[18]{r}{0.4\textwidth}
    \begin{minipage}{0.4\textwidth}

        \vspace*{-1cm}
        \begin{algorithm}[H]
            \DontPrintSemicolon%
            \SetInd{0cm}{.5cm}%
            \caption{}
            \label{alg2}
            \While{$\ttrue$}
            {
                $\semiplusbig_{j = 1}^N \{ i := j\};$\;
                \If{$\ell[i] = n$}{
                    $\ell[i] := w;$\;
                }
                \If{$\ell[i] = w$}{
                    \If{$y > 0$}{
                        $\semimult C_i;$ \; 
                        $y := y-1;$ \;
                        $\ell[i] := c;$\;
                    }\Else{
                        $\semimult W_i;$\;
                    }
                }
                \If{$\ell[i] = c$}{
                    $\semimult R_i;$ \; 
                    $y := y+1;$ \;
                    $\ell[i] := n;$\;
                }
            }
        \end{algorithm}
    \end{minipage}
\end{wrapfigure}
To see why we currently cannot express $\WLP$, 
consider \Cref{alg2} which describes
an operating system guaranteeing mutual exclusion. This algo\-rithm
 from
\cite{BatzGKKW22} (adapted from \cite{BK08}) 
handles $N$ processes that want to access a shared critical section which
 may only
be accessed by $y$ processes simultaneously. The status $\ell[i]$ of a selected process
$i$ can be either idle $(n)$, waiting $(w)$, or critical $(c)$. If the process $i$ is
idle, it becomes waiting. If it is waiting, one checks whether the shared section may be
entered $(y>0)$. Otherwise (if $y=0$), the process keeps on waiting. If process $i$ is
already in the critical section, it releases it and $y$ is updated.

In \cite{BatzGKKW22}, the natural language semiring and the alphabet $\Sigma = \{C_i, W_i,
R_i \mid 1 \leq i \leq N\}$ is used
to describe the different actions, i.e.,
if process $i$ enters the \underline{c}ritical section, \underline{w}aits, or
\underline{r}eleases the critical section, the corresponding branch is weighted by $C_i$, $W_i$, or $R_i$, respectively.
Now, $\WLP$ allows to reason
about the infinite paths represented by this loop.
More precisely, one can analyze the language of $\omega$-words produced by the loop via $\WLP$,
and show that there is a word like $W_2^{\omega}$ in the language, which means that process $2$ can wait infinitely long.
This disproves starvation freedom of \Cref{alg2}.

We can express this algorithm as an sARS similar to the one
in \Cref{ski example 1}.
As in \Cref{ex:starvation-freedom},  to analyze starvation freedom of such an algorithm, 
we can use the tuple semiring $\bigtimes_{i=1}^N \seminatext$.
Then the $i$-th component in a tuple describes how often the process $i$ has
entered the critical section.
In the steps with the weight $C_i$, the 
process $i$ enters
the critical section, so that we have to increase the value in the $i$-th component of the tuple.
The corresponding aggregator for such a step would be $e_i \semiplus v_1$, where $e_i$ denotes the tuple $(0, \ldots,0, 1, 0, \ldots, 0)$
where the $i$-th component is 1 and all other components are 0.
However, as in
\Cref{ex:starvation-freedom},
starvation freedom cannot be expressed with our current definition of $\semelem{a}$, but
we would have to analyze (bounds on) best-case reductions.

\report{

\section{Additional Theory and Proofs}\label{Additional Theory}

In the following, we state and prove the lemmas and theorems
that were mentioned throughout the paper.

\TupleSem*

\begin{proof}
    Since $\IS_i$ is a semiring for every $1 \leq i \leq n$,
    it follows directly that $\IS = \bigtimes_{i = 1}^{n} \IS_i$ is a semiring as well,
    as the addition and multiplication are defined pointwise.
    Moreover, we have:
    \begin{itemize}
        \item \emph{Naturally ordered}: Assume for a contradiction 
        that $\IS$ is not naturally ordered.
        Then, $\semileq_{\IS}$ is not antisymmetric, 
        i.e., there exist $(x_1, \ldots, x_n) \neq (y_1, \ldots, y_n)$ 
        such that $(x_1, \ldots, x_n) \semileq_{\IS} (y_1, \ldots, y_n)$ 
        and $(y_1, \ldots, y_n) \semileq_{\IS} (x_1, \ldots, x_n)$. 
        Since addition is pointwise, also the order is defined pointwise, 
        and hence, we have $x_i \semileq_{\IS_i} y_i$ 
        and $y_i \semileq_{\IS_i} x_i$ for all $1 \leq i \leq n$.
        Since $(x_1, \ldots, x_n) \neq (y_1, \ldots, y_n)$ 
        there must be a $1 \leq i \leq n$ such that $x_i \neq y_i$. 
        Then we have $x_i \neq y_i$, $x_i \semileq_{\IS_i} y_i$, 
        and $y_i \semileq_{\IS_i} x_i$, which is a contradiction to the antisymmetry of
        $\semileq_{\IS_i}$,  
        since $\IS_i$ is assumed to be naturally ordered.
        \item \emph{Complete Lattice}: Let $T \subseteq \IS$. 
        Its supremum $\semisup T$ is given by the supremum of each point, 
        i.e., $\semisup T = (\semisup T_1, \ldots, \semisup T_n)$, 
        where $T_i = \{x_i \mid (x_1, \ldots, x_n) \in T\}$. 
        Obviously, $(\semisup T_1, \ldots, \semisup T_n)$ is an upper bound for $T$: 
        For every $(x_1, \ldots, x_n) \in T$ there exist $u_1 \in \IS_1, \ldots, u_n \in \IS_n$ 
        with $x_i \semiplus_{\IS_i} u_i = \semisup T_i$ for all $1 \leq i \leq n$,
        and hence, $(x_1, \ldots, x_n) \semiplus_{\IS} (u_1, \ldots, u_n) 
        = (\semisup T_1, \ldots, \semisup T_n)$.
        Next, assume that there exists another upper bound $(w_1, \ldots, w_n) \in \IS$ 
        such that $(x_1, \ldots, x_n) \semileq_{\IS} (w_1, \ldots, w_n) 
        \semileq_{\IS} (\semisup T_1, \ldots, \semisup T_n)$. 
        Then, for all $1 \leq i \leq n$
        and all $x_i \in T_i$, 
        we have
        $x_i \semileq_{\IS_i} w_i \semileq_{\IS_i} \semisup T_i$. Since 
         $\semisup T_i$ is the supremum of $T_i$, this implies $w_i = \semisup T_i$. \qedhere
    \end{itemize}
\end{proof}

\SuffBoundOne*

\begin{proof}
    If $\fNF(a) = \top$ for $a \in \NFto$, then there exists a reduction tree
    consisting of a single node with label $a$ and weight $\top$, hence
    $\semelem{a} = \top$ and the wARS is not bounded.

    Now assume that the interpretation $\fNF$ is universally bounded by $\bound
    \in \IS\setminus\{\top\}$ and all aggregators $\aggrrule$ are selective. Let
    $\FT = (V,E)$ be a $(A,\to)$-reduction tree of finite depth. We show
    $\semantics{\FT}{v}{} \semileq \bound$ for all nodes $v \in V$ by induction on the
    height of $v$.   
    As usual, the \emph{height} of a node $v$ in a tree $\FT$ is the length
    of the maximal path from $v$ to a leaf in $\FT$.
    Hence, then we also have $\semelem{\FT} =\semantics{\FT}{r}{}  \semileq \bound$
    for the root $r$ of $\FT$,
    which is the node of maximal height in $\FT$.

    If $a_v \in\NFto$, then we have
    $\semantics{\FT}{v}{} = \fNF(a_v) \semileq \bound$.
    Otherwise, if $v$ is a leaf and $a_v \notin\NFto$, then
    $\semantics{\FT}{v}{} = \seminull \semileq \bound$.
    Finally, if  $v$ is an inner node and thus $a_v \notin\NFto$, then
    $\semantics{\FT}{v}{} =\aggr{a_v\to B}\left[
    \semantics{\FT}{w}{} \mid w \in vE \right] = 
    \semantics{\FT}{w}{}$ for some $w\in vE$, since the
    aggregator function is selective. Thus,
    $\semantics{\FT}{v}{} = \semantics{\FT}{w}{} \semileq \bound$ follows by the induction
    hypothesis.
\end{proof}

\SuffBoundTwo*

\begin{proof}
    For every $a \in A$ we prove that $\semelem{a} \neq \top$.
    Since the sARS is terminating, finitely
    non-deterministic, and finitely branching, there only exist finitely many trees in
    $\Phi(a)$. Thus,  their depth is bounded by
    some $D\in \IN$.

    Hence, it suffices to show that for every $a \in A$ and every $0 \leq n \leq D$, 
    there exists a constant $C_{a,n} \neq \top$ such
    that $\semantics{\FT}{v}{} \semileq C_{a,n}$ holds for all RTs $\FT \in \Phi(a)$
    and all their nodes $v$ of height $n$. 
    Then, $\semantics{\FT}{}{}  =\semantics{\FT}{r}{}
    \semileq C_{a,D}$ holds
    for the root $r$ of $\FT$. As this holds for 
    all 
    $\FT \in \Phi(a)$, we obtain  $\semelem{a} \semileq C_{a,D} \neq \top$.

    By induction on $n$, we now construct bounds   $C_{a,n} \neq \top$ such that
    $\semantics{\FT}{v}{} \semileq C_{a,n}$ holds for all
    nodes $v$ of height $n$.
    Since $\IS$ is a complete lattice, every finite set $\{s_1,\ldots,s_k\} \subset \IS
    \setminus \{ \top \}$  
    has an upper bound $\semisup \{s_1,\ldots,s_k\}$, and furthermore, 
    since $\IS$ has the extremal property, 
    we have $\semisup \{s_1,\ldots,s_k\} \semileq s_1 \semiplus \cdots \semiplus
    s_k \neq \top$.
    Let $\{a_1,\ldots,a_k\} \subseteq \NF_{\to}$ be the (finite) set of all normal forms 
    occurring in labels of nodes of trees in $\Phi(a)$,
    and let $C_{a,0} = \semisup \{\fNF(a_1),\ldots,\fNF(a_k)\}$.
    Since $\fNF(a_i) \neq \top$ for all $1 \leq i \leq k$, 
    we obtain $C_{a,0} \neq \top$.
    Then if $a_v \in\NFto$, we have
    $\semantics{\FT}{v}{} = \fNF(a_v) \semileq C_{a,0}$.
    Otherwise, if $v$ is a leaf and $a_v \notin\NFto$,
    then
    $\semantics{\FT}{v}{} = \seminull \semileq  C_{a,0}$.

    Now we regard the case $n \geq 1$.
    Here, $v$ is an inner node and thus, $a_v \notin\NFto$. Thus,
    $\semantics{\FT}{v}{} =\aggr{a_v\to B}\left[
    \semantics{\FT}{w}{} \mid w \in vE \right]$.
    By the induction hypothesis, we have 
    $\semantics{\FT}{w}{} \semileq C_{a,n-1}$
    and hence,  $\semantics{\FT}{w}{} \neq \top$ for all $w \in vE$.
    Since by assumption, aggregators are finite, the semiring $\IS$ has the extremal
    property, and $B$  is finite,
    $\aggr{a_v \to B}$ also satisfies the extremal property.
    Thus,
    \[
        \semantics{\FT}{v}{} 
        = \aggr{a_v \to B}\left[\semantics{\FT}{w}{} \mid w \in vE \right] 
        \neq \top.
    \] 
    Here, the second step follows from $\semantics{\FT}{w}{} \neq \top$ and the
    extremal property of the aggregator functions.
    Thus, for $n \geq 1$, we now define $C_{a,n} = \semisup
    \{ \semantics{\FT}{v}{} \mid \FT \in \Phi(a)$, $v$ is a node of
    $\FT$ at height $n \}$ and obtain  $C_{a,n} \neq \top$ since $\Phi(a)$ is a finite set of
    finite trees.
\end{proof}

\MonoContAggregator*

\begin{proof}
    The function $\semiplus$ is monotonic, as $\semileq$ is defined via addition.
    Monotonicity of $\semimult$ follows from distributivity:
    Let $s_1, s_2, t \in \IS$ with $s_1 \semileq s_2$.
    Then there exists a $u \in \IS$ with $s_1 \semiplus u = s_2$. Thus,
    $s_1 \semimult t \semileq s_1 \semimult t \semiplus u \semimult t = (s_1 \semiplus u) \semimult t = s_2 \semimult t$.
    The proof for the second argument is similar. Thus, finite sums and products are also monotonic.

    Infinite sums and products are monotonic as well:
    Let $\circ \in \{\semiplus, \semimult\}$ and
    let $[s_1, \ldots], [t_1, \ldots] \in \multiS$ be two infinite sequences
    such that $s_i \semileq t_i$ for all
    $i \in \IN$. 
    By monotonicity of $\circ$ for finite sums and products, we obtain
    $\bigcirc [s_1, \ldots, s_n] \semileq \bigcirc [t_1, \ldots, t_n]$
    for all $n \geq 1$. 
    Thus, \pagebreak[3]
    \begin{align*}
      & \bigcirc [s_1, \ldots] \\      
        = \quad& \semisup \left\{ \bigcirc \semifin \mid
        \semifin \text{ is a finite prefix of } [s_1, \ldots] \right\}\\
    \semileq \quad& \semisup \left\{ \bigcirc \semifin \mid
        \semifin \text{ is a finite prefix of } [t_1, \ldots] \right\}\\
        = \quad& \bigcirc [t_1, \ldots]\!
    \end{align*}
    Hence, infinite sums and products are monotonic as well.
    Finally, aggregator functions are monotonic since they are compositions of monotonic functions.

    Next, we consider continuity.
    Let $\semiplus$ and $\semimult$ be continuous and $\circ \in \{\semiplus, \semimult\}$.
    Again, we consider infinite sums and products first.
    Let $[s_1, \ldots] \in \multiS$
    be an arbitrary infinite sequence and
    let $K \subseteq \IS$. In the following, for two sequences $T, T' \in \multiS$, let $T
    \subseteq T'$ denote that $T$ is a prefix of $T'$.
    Then we obtain
    \begin{align*}
        & \semisup \{\bigcirc [s_1, \ldots, s_{j-1}, k, s_{j+1}, \ldots] \mid k \in K\}\\
        \downarrow&\small \text{(Definition)}\\
        = & \semisup \{\semisup \left\{ \bigcirc [s_1, \ldots, s_m] \mid
        [s_1, \ldots, s_m] \subseteq [s_1, \ldots, s_{j-1}, k, s_{j+1}, \ldots], \;
        m \in \IN \right\}\mid k \in K\}\\
        \downarrow&\small \; \text{(Combining the nested suprema to a supremum over a single set)}\\
        = & \semisup \Big\{\bigcirc [s_1, \ldots, s_m] \mid
        [s_1, \ldots, s_m] \subseteq [s_1, \ldots, s_{j-1}, k, s_{j+1}, \ldots], \;
        m \in \IN, k \in K\Big\}\\
        \downarrow&\small \; \text{(Partition set into sequences of length $m \geq j$ (containing $k$) or $m \leq j-1$ (not containing $k$))}\\
        = & \semisup \Big(\Big\{\bigcirc [s_1, \ldots, s_{j-1}, k, s_{j+1}, \ldots, s_m] \\
        & \qquad\qquad \Bigm\vert
        [s_1, \ldots, s_m] \subseteq [s_1, \ldots, s_{j-1}, s_{j+1}, \ldots], \;
        m \in \IN, m \geq j-1, k \in K\Big\} \\
        & \qquad \cup \Big\{\bigcirc [s_1, \ldots, s_m] \mid
        [s_1, \ldots, s_m] \subseteq [s_1, \ldots, s_{j-1}]\Big\} \Big)\\
        \downarrow&\small \; \text{(Creating nested suprema again)}\\
        = & \semisup \Big(\Big\{\semisup \{\bigcirc [s_1, \ldots,s_{j-1}, k, s_{j+1}, \ldots, s_m] \mid k \in K\}\\
        & \qquad\qquad \Bigm\vert [s_1, \ldots, s_m] \subseteq [s_1, \ldots, s_{j-1}, s_{j+1}, \ldots], \;
        m \in \IN, m \geq j-1\Big\} \\
        & \qquad \cup \Big\{\bigcirc [s_1, \ldots, s_m] \mid
        [s_1, \ldots, s_m] \subseteq [s_1, \ldots, s_{j-1}]\Big\} \Big)\\
        \downarrow&\small \; \text{(Continuity of finite sums and products)}\\
        = & \semisup \Big(\Big\{\bigcirc [s_1, \ldots, s_{j-1}, \semisup K, s_{j+1}, \ldots, s_m]\\
        & \qquad\qquad \Bigm\vert[s_1, \ldots, s_m] \subseteq [s_1, \ldots, s_{j-1}, s_{j+1}, \ldots], \;
        m \in \IN, m \geq j-1\Big\} \\
        & \qquad \cup \Big\{\bigcirc [s_1, \ldots, s_m] \mid
        [s_1, \ldots, s_m] \subseteq [s_1, \ldots, s_{j-1}]\Big\} \Big)\\
        \downarrow&\small \; \text{(Combining both sets of the union again)}\\
        = & \semisup \Big\{ \bigcirc [s_1, \ldots, s_m] \mid
        [s_1, \ldots, s_m] \subseteq [s_1, \ldots, s_{j-1}, \semisup K, s_{j+1}, \ldots], \;
        m \in \IN \Big\}\\
        \downarrow&\small \; \text{(Definition)}\\
        = & \bigcirc [s_1, \ldots, s_{j-1}, \semisup K, s_{j+1}, \ldots]\!
    \end{align*}
    Finally, aggregator functions are continuous since they are compositions of continuous functions.
\end{proof}

\BoundSoundAndComplete*

\begin{proof}
    We first show the ``if'' direction. To this end, we prove that
    for all reduction trees $\FT = (V,E)$ of finite depth and all nodes $v$ in $\FT$,
    we have 
    \begin{equation}
        \label{interpretationClaim}
        \mathfrak{e}(a_v) \quad \semigeq \quad \semantics{\FT}{v}{}.
    \end{equation}
    To see why this claim implies the ``if'' direction of the theorem, note that
    \eqref{interpretationClaim} implies
    $\mathfrak{e}(a)  \semigeq \semantics{\FT}{}{}$
    for all $a\in A$ and $\FT\in \Phi(a)$. This in turn implies $\mathfrak{e}(a) \semigeq
    \semelem{a}$. Since $\mathfrak{e}(a) \neq \top$ for all $a\in A$, the embedding
    provides upper bounds for the weights and proves boundedness of the wARS.

    We now prove the claim \eqref{interpretationClaim} by induction on the height of the
    node $v$.
    If $a_v \in\NFto$, then
    $\mathfrak{e}(a_v)\semigeq \fNF(a_v) = \semantics{\FT}{v}{}$.
    Otherwise, if $v$ is a leaf and $a_v \notin\NFto$, then
    $\mathfrak{e}(a_v)\semigeq \seminull = 
    \semantics{\FT}{v}{}$.
    Finally, if $v$ is an inner node and thus $a_v \notin\NFto$, then
    the induction hypothesis implies
    \begin{align*}
        & \forall w \in vE: \mathfrak{e}(a_w) \semigeq \semantics{\FT}{w}{} \\
        \Rightarrow \quad & \aggr{a_v\to B}\left[ \mathfrak{e}(a_w) \mid w \in vE \right]
        \semigeq
        \aggr{a_v\to B}\left[ \semantics{\FT}{w}{}  \mid w \in vE \right]
        \tag{as $\aggr{a_v\to B}$ is
        monotonic by \Cref{lem:mono-con-agg}}\\
        \Rightarrow \quad & \mathfrak{e}(a_v) \semigeq
        \aggr{a_v\to B}\left[ \mathfrak{e}(a_w) \mid w \in vE \right]
        \semigeq
        \aggr{a_v\to B}\left[ \semantics{\FT}{w}{}  \mid w \in vE \right]
        = \semantics{\FT}{v}{}.
        \tag{by the second requirement of the theorem}
    \end{align*}

    Now we show the ``only if'' direction. 
    Let $\semiplus$ and $\semimult$ be continuous, hence every aggregator function is continuous. 
    Suppose that the wARS $\wARS$ is
    bounded. Then we need to show that there exists an embedding $\mathfrak{e}$ that
    satisfies the conditions. We choose $\mathfrak{e}(a) = \semelem{a} $ for all $a \in
    A$. Since the wARS is bounded, this implies $\mathfrak{e}(a) \neq \top$ for all
    $a \in A$. Then if $a \in \NFto$, we have $\mathfrak{e}(a) = \semelem{a} = \semisup 
    \{ \semelem{\FT} \mid \FT \in \mathfrak{e}(a) \} =  \fNF(a)$, i.e., $\mathfrak{e}$ satisfies
    the first requirement of the theorem.
    If $a \notin \NFto$, then
    \begin{align*}
      \mathfrak{e}(a)
      & = \semelem{a}
      = \semisup_{\FT \in \Phi(a)} \semelem{\FT}
      = \semisup_{\FT \in \Phi(a)}  \semantics{\FT}{v}{}
      \tag{where $v$ is root of $\FT$} \\
     & \semigeq \semisup_{\FT \in \overline{\Phi(a)}}  \semantics{\FT}{v}{}
      \tag{$\overline{\Phi(a)} \subseteq \Phi(a)$ are the RTs where the first reduction
        is $a\to B$} \\
      & = \semisup_{\FT = (V,E) \in \overline{\Phi(a)}}
          \aggr{a \to B} \left[\; 
          \semantics{\FT}{w}{} \mid w\in vE \;\right] \\
      &  = \aggr{a \to B} \Big[
          \semisup_{\FT = (V,E) \in \overline{\Phi(a)}} \semantics{\FT}{w}{} \;\Big|\; w\in vE \;\Big]
      \tag{since $\aggrrule$ is continuous} \\
      & = \aggr{a \to B} \left[\;
          \semisup_{\FT \in \Phi(b)} \semantics{\FT}{}{}
          \mid b\in B \;\right]
      = \aggr{a \to B} \left[\;
          \semelem{b} \mid b\in B \;\right] \\
      & = \aggr{a \to B} \left[\;
          \mathfrak{e}(b) \mid b\in B \;\right]. \tag{due to our choice of $\mathfrak{e}$}
    \end{align*}
    This holds for every possible reduction $a\to B$,
    so $\mathfrak{e}$ also satisfies the second requirement of the theorem.
\end{proof}

\Approximating*

\begin{proof}
    We prove that for all reduction trees $\FT = (V,E)$ of finite depth,
    all $n', n \in \IN$, and all nodes
    $v$ of
    $\FT|_{n'}$, we have:
    \[ n' \leq n \quad \text{ implies } \semantics{\FT|_{n'}}{v}{} \semileq
    \semantics{\FT|_n}{v}{}.\]
    The claim is proved by induction on the height of $v$ in $\FT|_n$.
    If $v$ is a leaf of $\FT|_n$, then it is also a leaf of $\FT|_{n'}$
    and thus, the claim is clear, because then $\semantics{\FT|_n}{v}{}$ does not depend on
    $n$.

    Thus, let  $v$ be an inner node of $\FT|_n$
    and thus, $a_v \notin \NFto$. Moreover, let
    $a_v \to B$ with $B = [a_w \mid w \in vE]$. If
    $v$ is a leaf of $\FT|_{n'}$, then
    $\semantics{\FT|_{n'}}{v}{} =
    \seminull$, which proves the claim. Otherwise, we have
    \begin{align*}
    & \semantics{\FT|_{n'}}{v}{} \\
    = \quad &
    \aggr{a_v\to B}\left[\semantics{\FT|_{n'}}{w}{}\mid w \in vE \right]\\
    \semileq\quad &  \aggr{a_v\to B}\left[\semantics{\FT|_{n}}{w}{}\mid w \in vE \right]\\
    =\quad&\semantics{\FT|_{n}}{v}{}
    \end{align*}
    Here, in the but-last line we use that
    $\aggr{a_v\to B}$ is monotonic by \Cref{lem:mono-con-agg} and that the induction hypothesis implies
    $\semantics{\FT|_{n'}}{w}{} \semileq \semantics{\FT|_{n}}{w}{}$.

    As the sARS is deterministic, $\Phi(a)$ just consists of $\FT$ and its prefixes. This
    implies
    $\semisup \{\semantics{\FT|_{n}}{}{} \mid n \in \IN\} = \semelem{a}$.
\end{proof}

\Loops*

\begin{proof}
    Since $v_0 \neq r$, we have $a \notin\NFto$.
    Hence, 
    by replacing the leaf $v_0$ with the finite tree $\F{T}$ repeatedly,
    one obtains an infinite sequence of finite trees $\FT = \F{T}_1, \F{T}_2, \ldots$ with
    $\semantics{\FT_n}{}{} \semigeq \semiplusbig_{i = 1}^{n} t$ for all $n \in \IN$.

    To prove this by induction on $n$, note that $a \notin\NFto$
    implies $\semantics{\FT_1}{}{} = \semantics{\FT}{}{} = \PP_{v_0}(\F{T})(\seminull) \semigeq \seminull
    \semiplus t = t$. For $n > 1$, we have
    $\semantics{\FT_n}{}{} =
    \PP_{v_0}(\FT)(\semantics{\FT_{n-1}}{}{})
    \semigeq  \semantics{\FT_{n-1}}{}{} \semiplus t
    \semigeq  \semiplusbig_{i = 1}^{n-1} t \semiplus t
    = \semiplusbig_{i = 1}^{n} t$ by the
    induction hypothesis and the monotonicity of $\semiplus$ (\Cref{lem:mono-con-agg}).

    Thus, for the infinite extension of this construction we have
    \begin{align*}
    \semelem{a} \quad =& \quad \semisup \{ \semantics{\FT}{}{} \mid \FT \in \Phi(a) \}\\
    \semigeq& \quad \semisup \{ \semantics{\FT_1}{}{}, \semantics{\FT_2}{}{}, \ldots \} \tag{by
        \Cref{cor:approx} }\\
    \semigeq& \quad \semiplusbig_{i = 1}^{\infty} t \quad = \quad \semitop \tag{by \Cref{def:infinite_operations}}
    \end{align*}
\end{proof}

\begin{lemma}[Complete Semilattice Versus Complete Lattice] \label{lem:semilattice}
    Let $(S,\leq)$ be a partially ordered set where the supremum exists
    for every subset, i.e., $\semisup T \in S$ for all $T\subseteq S$.
    Then the infimum exists as well for every subset, i.e.,
    $\semiinf T \in S$ for all $T\subseteq S$.
\end{lemma}

\begin{proof}
    Given an arbitrary set $T\subseteq S$, we define the set
    $\text{lb}(T) = \{ \underline{t}\in S \mid \underline{t}~\leq~t
    \text{ for all } t~\in~T \}$ that consists of all \emph{lower bounds} of $T$
    and $\text{ub}(T) = \{ \overline{t}\in S \mid t \leq \overline{t}
    \text{ for all } t\in T \}$ that consists of all \emph{upper bounds} of $T$.
    Then $u = \semisup \text{lb}(T)$ exists, and
    we show that $u$ is the greatest lower bound of $T$.
    \begin{enumerate}
        \item For all $t \in T$ and every $\underline{t} \in \text{lb}(T)$, we have $\underline{t} \leq t$, 
        and thus, $t$ is also an upper bound on $\text{lb}(T)$, i.e., $t \in \text{ub}(\text{lb}(T))$, 
        hence $u \leq t$, since $u$ is the least such upper bound. 
        This implies that $u$ is a lower bound of $T$, i.e., $u \in \text{lb}(T)$.

        \item Given $u' \in \text{lb}(T)$, we have $u' \leq u$ since $u \in \text{ub}(\text{lb}(T))$.
        Hence, $u$ is the greatest lower bound of $T$.
    \end{enumerate}
\end{proof}

\begin{lemma}[Uncountably Many Reduction Trees] \label{lem:nondeterminism}
    There exists a finitely non-deterministic 
    ARS $(A,\to)$ which leads to uncountably many
    reduction sequences starting with a fixed initial object $a \in A$.

    Moreover, there exists a wARS $\wARS$ such that there are uncountably many
    reduction trees for a fixed initial object $a \in A$,
    and each such reduction tree has a different weight.
\end{lemma}

\begin{proof}
    For the first part of the lemma,
    consider the ARS $(\{0,1\}, \to)$ with $0 \to 0$, $0 \to 1$, $1 \to 0$, and $1 \to 1$.
    Then the infinite reduction sequences correspond to all infinite sequences of $0$'s and
    $1$'s. Hence, there is a bijection between the set of infinite reduction sequences starting with
    $0$ and the set $2^{\IN}$ which is known to be uncountable.

    For the second part, consider the sARS 
    $(A, \to) = (\{0, 1, *\}, \to)$ where the relation $\to$ is given by the following rules:
    $$ 0 \to [0,*], \quad 0 \to [1,*], \quad 1 \to [0,*], \quad 1 \to [1,*].$$
    Then, the reduction trees with root $0$ have the form:

    \begin{center}
    \begin{tikzpicture}[level distance=20pt,level 1/.style={sibling distance=50pt},level 2/.style={sibling distance=20pt}]
        \tikzstyle{myRect}=[thick,draw=black!100,fill=white!100,minimum size=4mm, shape=rectangle]
        \tikzstyle{trueRect}=[thick,draw=black!100,fill=lipicsYellow!100,minimum size=4mm, shape=rectangle]

        \node {0}
            child {node {0}
                child {node[label=below:\vdots] {0}}
                child {node {*}}
            }
            child {node {*}};
    \end{tikzpicture}
    \begin{tikzpicture}[level distance=20pt,level 1/.style={sibling distance=50pt},level 2/.style={sibling distance=20pt}]
        \tikzstyle{myRect}=[thick,draw=black!100,fill=white!100,minimum size=4mm, shape=rectangle]
        \tikzstyle{trueRect}=[thick,draw=black!100,fill=lipicsYellow!100,minimum size=4mm, shape=rectangle]

        \node {0}
            child {node {0}
                child {node[label=below:\vdots] {1}}
                child {node {*}}
            }
            child {node {*}};
    \end{tikzpicture}
    \begin{tikzpicture}[level distance=20pt,level 1/.style={sibling distance=50pt},level 2/.style={sibling distance=20pt}]
        \tikzstyle{myRect}=[thick,draw=black!100,fill=white!100,minimum size=4mm, shape=rectangle]
        \tikzstyle{trueRect}=[thick,draw=black!100,fill=lipicsYellow!100,minimum size=4mm, shape=rectangle]

        \node {0}
            child {node {1}
                child {node[label=below:\vdots] {0}}
                child {node {*}}
            }
            child {node {*}};
    \end{tikzpicture}
    \begin{tikzpicture}[level distance=20pt,level 1/.style={sibling distance=50pt},level 2/.style={sibling distance=20pt}]
        \tikzstyle{myRect}=[thick,draw=black!100,fill=white!100,minimum size=4mm, shape=rectangle]
        \tikzstyle{trueRect}=[thick,draw=black!100,fill=lipicsYellow!100,minimum size=4mm, shape=rectangle]

        \node {0}
            child {node {1}
                child {node[label=below:\vdots] {1}}
                child {node {*}}
            }
            child {node {*}};
    \end{tikzpicture}
    \begin{tikzpicture}[level distance=20pt,level 1/.style={sibling distance=50pt},level 2/.style={sibling distance=20pt}]
        \tikzstyle{myRect}=[thick,draw=black!100,fill=white!100,minimum size=4mm, shape=rectangle]
        \tikzstyle{trueRect}=[thick,draw=black!100,fill=lipicsYellow!100,minimum size=4mm, shape=rectangle]

        \node {\ldots};
    \end{tikzpicture} 
    \end{center}
\noindent
    So each infinite tree represents an infinite string of $0$'s and $1$'s and
    hence, we have $|\Phi(0)| = |2^{\IN}|$.

    To give a corresponding wARS where all these trees have different weight, 
    consider the formal languages semiring $\semilang$ over $\Sigma = \{0,1\}$,
    the interpretation of the normal forms $\fNF(*) = \{\varepsilon\}$,
    and the aggregator functions
    \begin{align*}
        \mathsf{Aggr}_{0 \to [0,*]}(x, x_{*}) = \mathsf{Aggr}_{0 \to [1,*]}(x, x_{*}) 
        &= \left( \{0\} \semimult_{\semilang} x \right) \semiplus_{\semilang} x_{*} = \{0w \mid w \in x\} \cup x_{*}\\
        \mathsf{Aggr}_{1 \to [0,*]}(x, x_{*}) = \mathsf{Aggr}_{1 \to [1,*]}(x, x_{*}) 
        &= \left( \{1\} \semimult_{\semilang} x \right) \semiplus_{\semilang} x_{*} = \{1w \mid w \in x\} \cup x_{*}\!
    \end{align*}
    For any reduction tree $\FT$, now the set $\semantics{\FT}{}{}$ consists of all finite prefixes of the infinite string represented by the tree.
    Hence, we have $\semantics{\FT}{}{} \neq \semantics{\FT'}{}{}$ for all $\FT, \FT' \in \Phi(0)$ with $\FT \neq \FT'$.
    In other words, we have a wARS where $\Phi(0)$ contains uncountably many trees 
    that all have a different weight.
\end{proof}
}
\end{document}